\documentclass[letterpaper,12pt]{article}
\usepackage{amsmath, mathtools}
\usepackage{amssymb}
\usepackage{amsthm}
\usepackage{subfig}
\usepackage{graphicx}
\usepackage{enumitem}
\usepackage{verbatim}
\usepackage{orcidlink}
\usepackage{color}
\usepackage{url}
\usepackage{longtable}
\newtheorem{theorem}{Theorem}

\newtheorem{lemma}{Lemma}

\newtheorem{prop}{Proposition}

\def\PP{\mathbf{P}}     
\def\EE{\mathbf{E}}     
\def\Var{\textnormal{Var}}
\def\Li{\textnormal{Li}}
\usepackage[authoryear,round,numbers]{natbib}
\definecolor{Red}{rgb}{1,0,0}
\definecolor{Blue}{rgb}{0,0,1}
\definecolor{Pink}{rgb}{0,0,0}

\def\pink{\color{Pink}}

\newcommand{\rev}[1]{{\pink #1}}

\date{\today}

\author{
Brandon Legried\footnote{
School of Mathematics, Georgia Institute of Technology, Atlanta, GA.  Corresponding author.  Email:  blegried3@math.gatech.edu.
ORCID [0000-0003-2016-753X,\orcidlink{0000-0003-2016-753X}] }
}

\title{Anomaly zones for uniformly sampled gene trees under the gene duplication and loss model}

\begin{document}

\maketitle

\begin{abstract}
Recently, there has been interest in extending long-known results about the multispecies coalescent tree to other models of gene trees.  Results about the gene duplication and loss (GDL) tree have mathematical proofs, including species tree identifiability, estimability, and sample complexity of popular algorithms like ASTRAL.  Here, this work is continued by characterizing the anomaly zones of uniformly sampled gene trees.  The anomaly zone for species trees is the set of parameters where some discordant gene tree occurs with the maximal probability.  The detection of anomalous gene trees is an important problem in phylogenomics, as their presence renders effective estimation methods to being positively misleading.  Under the multispecies coalescent, anomaly zones are known to exist for rooted species trees with as few as four species.

The gene duplication and loss process is a generalization of the generalized linear-birth death process to the rooted species tree, where each edge is treated as a single timeline with exponential-rate duplication and loss.  The methods and results come from a detailed probabilistic analysis of trajectories observed from this stochastic process.  It is shown that anomaly zones do not exist for rooted GDL balanced trees on four species, but \rev{do} exist for rooted caterpillar trees, as with \rev{the} multispecies coalescent.
\end{abstract}

\section{Introduction}

The reconstruction of phylogenetic trees often begins with the analysis of molecular sequences of existing species.  Probabilistic and computational methods are used to establish rigorous convergence results as the amount of data goes to infinity.  The phylogenomic framework is a two-step approach.  First, molecular sequences are used to reconstruct gene phylogenies that depict the evolution of a locus within the genome, and the existing computational technology allows practitioners to collectively estimate many gene trees.  A simplifying assumption in the phylogenomic approach is that locus sequences are disjoint so that their evolutionary trajectories are roughly independent.  Then, the species tree originating the data is constructed from the many independent gene trees.  \rev{The gene trees might be assumed to be computed without gene tree estimation error (GTEE), which can happen if the sequences are relatively short.}

Even with these simplifying assumptions, species tree estimation is confounded by gene tree heterogeneity.  \rev{Heterogeneity is particular problematic for concatenation-based methods, as the species tree for the entire concatenated sequence can disagree with gene trees for particular loci, \cite{RochSteel:15}.} Common sources of heterogeneity include incomplete lineage sorting (ILS) \cite{rannala2003bayes}, horizontal gene transfer (HGT) \cite{roch12lateral}, and gene duplication and loss (GDL) \cite{ArLaSe:09}.  Many theoretical results, positive and negative, have been established when the only source of heterogeneity is ILS, see \cite{degnan2006}, \cite{allman2011identifying}, and \cite{astral}.  Incomplete lineage sorting is modeled by the \rev{multispecies} coalescent (MSC) model.

If the gene trees are assumed independent \rev{of each other}, then the \rev{``democratic vote''} estimator \rev{finds the species tree with the highest probability by counting the number of times each branching pattern appears in the list of gene trees.  As more independent gene trees are accumulated, the gene tree with the highest probability obtains the most votes almost surely.}  \rev{As in \cite{degnan2006}, this estimate of the species tree is simply the gene tree topology that occurs most often.}  Even under these ideal assumptions, there exist species trees for which the \rev{democratic vote of MSC gene trees} is positively misleading.  Such rooted trees exist when the number of species is as few as four.  Similarly, the \rev{democratic vote estimate} can be positively misleading for some unrooted trees with as few as five species.  Species trees are in the \textit{anomaly zone} if the gene tree with maximum probability is discordant from the species tree (\cite{degnan2006}).  However, by using supertree methods such as ASTRAL (\cite{astral}) on unrooted quartets, any unrooted species tree can be consistently estimated in \rev{a polynomial number of species and polynomial number of gene trees.  The ASTRAL suite has found extensive usage across many biological datasets.}  Finite sample guarantees have also been developed, see \cite{Shekhar_2018}.  \rev{This type of result assumes some level of error tolerance $\epsilon$, then provides a minimum number of genes that are required to obtain a provable amount of error below the tolerance level.}

Much less is known about species tree estimation in the presence of GDL.  Recently, using probabilistic and cominbatorial arguments, it was shown that \rev{the ``democratic vote winner''} is a consistent estimator of unrooted quartets under GDL (\cite{legriedGDL2021}), so ASTRAL is also consistent when the input data are GDL trees rather than MSC trees.  \rev{A parallel result holds for rooted triples through a nearly identical analysis.} The result was generalized further when gene trees come from the DLCoal model, where GDL and ILS occur simultaneously, see \cite{eulenstein2020unified} and \cite{hill2022species}.  \rev{Both of these results are counter-intuitive, as the ASTRAL pipeline was not developed with the GDL model in mind.} Finite sample guarantees for ASTRAL have been proven, showing a sufficient amount of data needed to obtain high probability results \cite{hill2022species}.  

In this paper, the distribution of gene trees is described further for gene trees generated under GDL.  With this further information, we describe when anomaly zones can exist for gene trees generated under GDL for rooted species trees on either three or four species.  As with anomalous gene trees in the multispecies coalescent model, the lengths of interior edges of the species tree are important.  \rev{As the interior branch lengths in the species tree grow to infinity, the probability that the gene tree topology coincides with that of the species tree goes to $1$.  The discordant gene trees have less probability.  Similarly for GDL}, species trees with longer interior edges have lower probabilities of discordant gene trees.  However, the parameters governing birth and death are also relevant.  \rev{As observed in \cite{hill2022species},} when the per-capita birth rate is high, the number of edges is high and the signal emitted by the species tree diminishes.  Conversely, when the birth rate is 0, \rev{every discordant gene tree has probability zero, for any setting of branch lengths in the species tree.}  Similar effects occur when the death rate is high enough to prevent excessive branching in the GDL process, \rev{but explicit quantitative results are required to understand this effect.}  This paper provides results that aid in intuiting the connection between the birth and death rates and gene tree discordance, but the focus is on the \rev{number of copies} in the ancestral population rather than the birth and death rates themselves.  The main results apply to any choice of birth and death rates and species trees with three or four leaves.

The format of this paper is as follows.  In Section 2, we make precise the definition of anomalous gene trees in the GDL context.  In Section 3, the population process is recalled, and a short result relating the arithmetic mean and quadratic mean of independent progenies is proved.  In Section 4, the relationship between the population process and the species tree is investigated.  In Section 5, it is shown that the rooted balanced quartet has no anomaly zones, in contrast to the MSC model.  In Section 6, it is shown that the rooted caterpillar quartet may or may not have anomaly zones, but that only balanced gene trees could possibly be anomalous.  

\section{Problem and model}

In this section, we describe the model and state the results.

Let $S$ be the collection of species and $\sigma = (\mathcal{T},\mathbf{f})$ be the species tree with topology $\mathcal{T} = (V,E)$ and branch lengths $\mathbf{f} = \{f_e\}_{e \in E}.$  The species tree topology $\mathcal{T}$ contains only the vertices $V$ and the edges $E = \{(u,v)\}_{u,v \in V}.$ The problem is to estimate $\mathcal{T}$ from a collection $\mathcal{G} = \{t_\iota\}_{\iota=1}^{K}$ of $K$ multi-labelled gene tree topologies.  A \textit{gene tree} is a depiction of the parental lineages of a gene or multiple gene copies from individuals across several species.  The gene tree is \textit{multi-labelled} in that each leaf is assigned exactly one label from $S$, multiple leaves may be labelled with the same species.  By contrast, a gene tree is \textit{single-labelled} if no species is used more than once in the labelling of leaves.  The process of gene duplication may create multiple copies of a gene within the same individual in a species. We refer to these duplicated genomic segments as \textit{gene copies} and we refer to collections of gene copies from different unrelated genes as \textit{gene families}.  Gene families depict the joint ancestral history of these related gene copies. In contrast, a \textit{species tree} is a depiction of the evolutionary relationships of a group of species.

Multiple leaves of a gene tree may be labelled by the same species; this corresponds to observing multiple paralogous copies (i.e., which have arisen from gene duplication) of a gene within the genome. In practice, gene trees are estimated from the molecular sequences of the corresponding genomic segments using a variety of phylogenetic reconstruction methods, see \cite{SempleSteel:03,felsenstein2003inferring,gascuel2005mathematics,yang2014molecular,steelbook2016,warnow2017book}. Here, we assume that gene trees are provided \textit{without estimation error} for a large number of gene families.
Our main modeling assumption is that these gene trees have been drawn independently from a distribution for which the species tree $\sigma$ is taken as a fixed parameter. We define the model more precisely next.

\paragraph{Model.} The gene trees are assumed to be independent and identically distributed. There are $\kappa$ gene trees given.  The process for generating a gene tree under GDL proceeds in two steps.  The process is repeated independently for each $\iota \in \{1,\dots,\kappa\}$.

Starting with a single ancestral copy of a gene at the root of $\mathcal{T}$, a tree is generated by a \rev{top-down birth-death process \textit{within} the species tree, ~\cite{ArLaSe:09}}.  On each edge in $\mathcal{T}$, each gene copy independently evolves.  It duplicates at exponential rate $\lambda \geq 0$ and is lost at exponential rate $\mu \geq 0$.  Each gene copy that survives to a speciation vertex in $\mathcal{T}$ undergoes a bifurcation into two child copies, one for each descendant edge of $\mathcal{T}$.  The process continues inductively up to the present.  The gene tree is then pruned of lost copies.  (These lineages cannot be observed.) Species labels are assigned to each leaf from $S$.  Each bifurcation in the gene tree arises through duplication or speciation.  An intuitive way to view this process is that $\sigma$ is depicted using a fat tree.  This tree constrains the linear birth-death process so that it contains a skinny tree.  A sampled realization of the GDL process on a tree with three species is given in Figure \ref{fig:GDL_three_leaves}.  \rev{Each branch of the gene tree is associated to an evolving sequence.  The results of this paper hold regardless of the choice of evolutionary model, so we have no need to specify any here.  Common choices would be the Jukes-Cantor model or general time-reversible (GTR) model, see \cite{tavare1986seq} for a full definition and discussion.}

\begin{figure}

\centering
\includegraphics[scale=0.3]{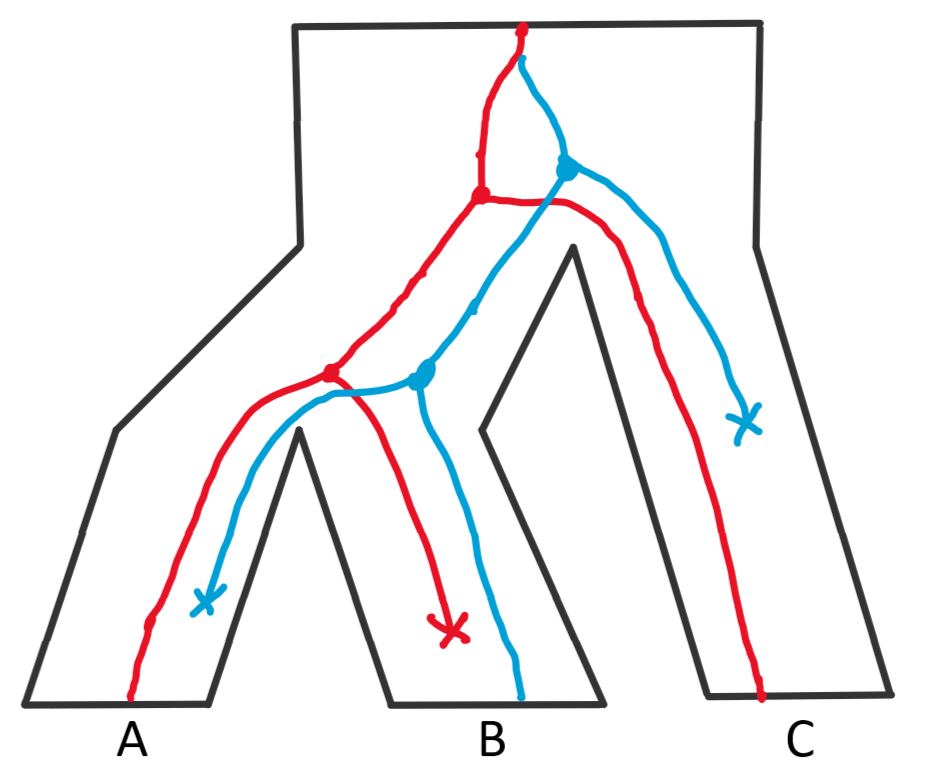}

\caption{The fat tree is the species tree on the species $A$,$B$,$C$.  The skinny tree (colored red and blue) indicates a possible gene tree where no species goes extinct. Note the topology of the pruned tree (nodes labelled ``X'' died so they are removed from the tree) puts $A$ and $C$ more closely related than $A$ and $B$, showing how discordance can arise under gene duplication and loss.}
\label{fig:GDL_three_leaves}
\end{figure}

It is assumed that there is one and only one copy from each species in the gene tree topology $t = t_{\iota}$.  Such gene trees generated by GDL have been called ``pseudoorthologs'', e.g. \cite{smith2022}.  In the next few lemmas, we analyze the likelihood function of $t$ under these assumptions.   Throughout, the species tree $\sigma = (\mathcal{T},\mathbf{f})$ is assumed to have no more than four leaves.  The edge lengths $\mathbf{f}$ are set so that the tree is ultrametric, meaning all leaves have the same distance to the root.

Gene trees under GDL are multi-labelled, in contrast to singly-labelled MSC gene trees, but the singly-labelled GDL gene trees described in the previous paragraph are still highly useful in biological problems.  In the estimation methods for species tree reconstruction (\cite{rabiee2019multi}, \cite{legriedGDL2021}, \cite{yan2022}), multi-labelled gene trees are pre-processed into a collection of singly-labelled trees.  We will need only trees generated for ASTRAL-one in this paper.  Conditioned on no species going extinct in the \rev{multispecies} linear birth-death process (this conditioning is acknowledged by $\mathbf{P}'$), ASTRAL-one selects one gene copy from each species in the gene family uniformly at random and removes the other lineages from the gene tree.  The result is a singly-labelled tree $(u(t),\mathbf{f}(t))$, which we denote $U(t)$.  Realizations $U(t)$ are called \textit{uniformly sampled gene trees}.  These singly-labelled gene trees may then be evaluated to find support for different hypotheses of the species tree.  In this paper, we describe the distribution of uniformly sampled gene trees, as they have shown to be informative to estimating the species tree.  \rev{Alternate methods to ASTRAL-one are proposed in \cite{legriedGDL2021} and \cite{yan2022} and should be considered, even though they are not studied here.  ASTRAL-only is a method that takes only singly-labelled gene trees as input, meaning the input consists entirely of orthologs and pseudoorthologs.  ASTRAL-all or ASTRAL-multi takes each gene tree and extracts \textbf{all} singly-labelled trees obtained by selecting a copy from each species.  Another twist on ``one'' and ``all'' is to construct estimates without the whole gene tree -- that is, use standard methods to reconstruct phylogenies from sampled copies only.  Other methods for processing multi-labelled gene trees are given in \cite{molloy2020} and \cite{zhang2020}.}

\section{Structured numbers of copies under the GDL process}

The GDL process is a generalization of the linear birth-death process, and some of the basic results are recalled here.  \rev{The development of the linear birth-death process started in mathematics papers such as \cite{kendall1948}.}  In the species tree with stem edge of length $s$, let $N_s$ be the \rev{number of copies} at time $s$.  The probability mass function of $N_s$ is denoted $p_i(s) = \PP(N_s = i).$  Then for any $i \geq 0$:  $$\PP(N_s = i) = \begin{dcases}
\frac{\mu}{\lambda}q(s) & \textnormal{if} \ i = 0 \\
e^{-(\lambda - \mu)s}\left(1 - p_0(s)\right)^2 q(s)^{i-1} & \textnormal{if} \ i > 0,
\end{dcases}$$ where $$q(s) = \begin{dcases}
\frac{\lambda - \lambda e^{-(\lambda - \mu)s}}{\lambda - \mu e^{-(\lambda - \mu)s}} & \textnormal{if} \ \lambda \ne \mu \\
\frac{\lambda s}{1+\lambda s} & \textnormal{if} \ \lambda = \mu.\end{dcases}$$ The derivation of this \textit{modified geometric} random variable is provided in Chapter 9 of \cite{steelbook2016}.  \rev{It is modified geometric in the sense that given that the population does not go extinct at or before than time $s$ (i.e. $N_s > 0$), the population size at time $s$ is a non-modified geometric random variable with parameter $1-q(s)$.  That is, $$\mathbf{P}'(N_s = i) = \mathbf{P}(N_s = i|N_s > 0) = (1-q(s))q(s)^{i-1}$$ for any $i > 0$.  The canonical game for a geometric random variable is that $N_s$ represents the required number of independent trials to observe the first failure when the probability of success is $q(s)$.  Selected plots of $\textnormal{log}[q(t)]$ are provided in Figure \ref{fig:qplot}.  We choose to plot the log because the individual values of $q(t)$ are very small.  The log has the appearence of a logarithm function, meaning that $q(t)$ is approximately a linear function of $t$, whose slope is determined by $\lambda$ and $\mu$.  Note there is symmetry in $\lambda$ and $\mu$ in the sense that $q(t)$ is increasing in $t$, regardless of the sign.  Similar symmetric observations are observed in genetics, e.g. \cite{hill2022species} and \cite{legried2021birthdeath}.

\begin{figure}%
    \centering
    \subfloat[\centering]{{\includegraphics[width=6cm]{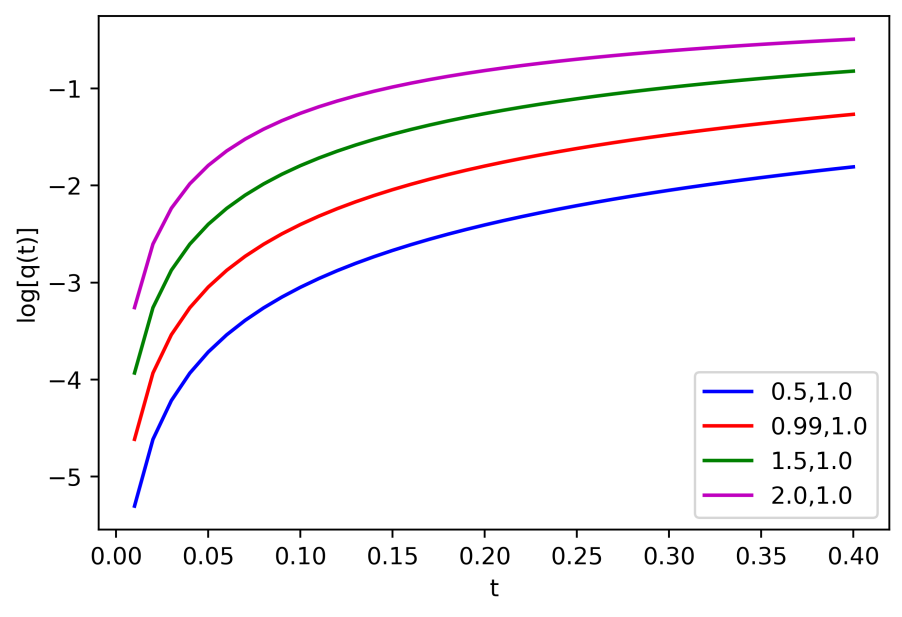} }}%
    \qquad
    \subfloat[\centering]{{\includegraphics[width=6cm]{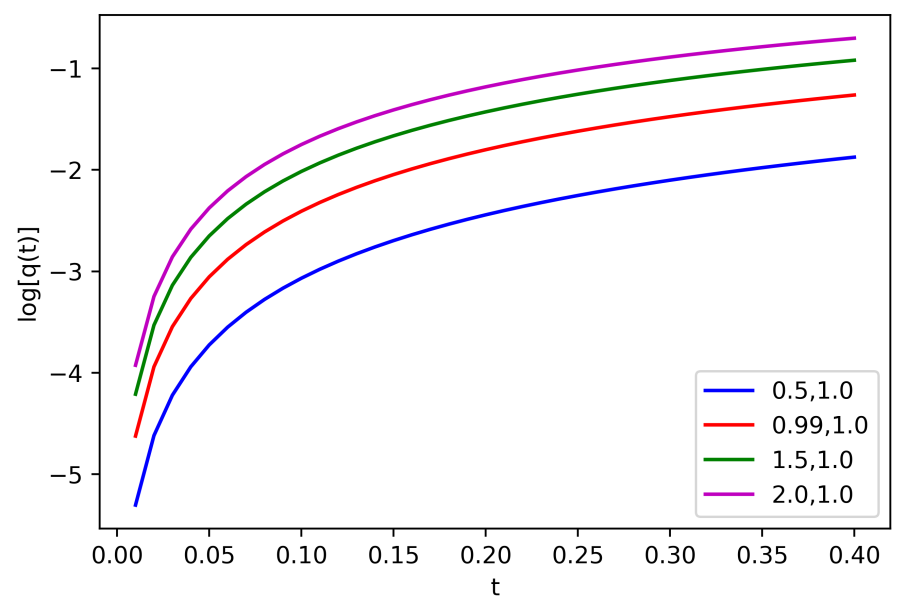} }}%
    \caption{(a) The plot of $\textnormal{log}[q(t)]$ as a function of $t$, for selected values of $\lambda$ and $\mu$.  The difference $\lambda - \mu$ is kept constant.  (b) The plot of $\textnormal{log}[q(t)]$ as a function of $t$, for selected values of $\lambda$, but keeping $\mu$ fixed.}%
    \label{fig:qplot}%
\end{figure}

There is a natural interpretation of $q(s)$ as the chance of success in a single trial.  In the simple case with no death (i.e. $\mu = 0$), the value simplifies to $q(s) = 1 - e^{-\lambda s}$, which is the probability that a single individual alive at time $0$ gives birth to at least one new offspring at or before time $s$.  Because the births occur in exponential time, the distribution of the total progeny $N_s$ assigns equal probability density to every tree relating the offspring.  As a result, any new offspring generated may be viewed as having an equal chance of having some offspring by time $s$ as if it existed at time $0$.  Eventually, some new offspring will fail to give birth to any further offspring, terminating the process.  In this setting, ``failure'' corresponds to a single individual having no new offspring.  For general $\mu > 0$, the interpretation is similar, though the success probability $1-q(s)$ additionally incorporates assumptions of survival.}  

\rev{We will utilize all components of the GDL model to perform species tree estimation.  As stated in the Introduction, the primary focus will be the observed \rev{numbers of copies} at speciation nodes in the species tree.  We review known results and provide new technical results in the Appendix.  The accumulation of these results is presented in the following Theorem.

The basic building block is the \textit{chain} species tree, shown in Figure \ref{figure:GDL2node}. The tree $\sigma$ has branching pattern $\mathcal{T}$ with a root node $R$, followed by a single interior node $I$, and one child $J$.  Any chain can be generalized to a species tree with branching (i.e. speciation) by appending a subtree to interior nodes like $I$. An example is shown in Figure \ref{figure:GDLuniformcopy}.  Let $N_I$ be the number of surviving copies to $I$.  If $N_I$ is known, then for each $j \in \{1,...,N_I\}$: one can let $N_{I,j}$ be the number of surviving copies to $J$ that descend specifically from $j$.  Of interest is the relationship between the arithmetic mean and geometric mean of the $N_{I,j}$ when the weight of the edge from $I$ to $J$ is $f > 0$.  For that, we use the moment generating function $M_f(\tau) = \mathbf{E}'[e^{\tau N_{I,1}}|N_I]$ of the progeny obtained by a single individual on the branch between $I$ and $J$.  We can use the moment generating function to make precise the relationship between the arithmetic mean and the quadratic mean of independent and identically distributed copies of the modified geometric.  \begin{theorem}
    Conditioned on the value of $N_I \geq 1$, let $N_{I,j}$ be the number of surviving progeny to $J$ for each $j \in \{1,2,...,N_I\}$.  Then \begin{equation*}
    \EE'\left[\frac{\sum_{j=1}^{N_I}N_{I,j}^2}{\left(\sum_{j=1}^{N_I}N_{I,j}\right)^2} \bigg| N_I \right] = N_I\int_{\tau=0}^{\infty}\tau M_{f}''(-\tau) M_{f}(-\tau)^{N_I - 1} \ d\tau.
\end{equation*} 
\end{theorem} The specific form of $M_f(\tau)$ is outlined in the Appendix.  Unfortunately, this integral is difficult to compute analytically.  However, numerical methods could be useful to characterizing the expectation of this ratio.  This result could conceivably be used to give exact expressions for the probability of each tree topology, and an implicit formula is given in Proposition \ref{prop:3}.  Explicit tools are developed for rooted trees with three species in the Appendix, though the results do not find application beyond those of \cite{legriedGDL2021}.  In the rest of this paper, we focus on deriving simpler one-sided results to obtain relevant results for rooted trees with four species.

\begin{figure}
    \centering
    \includegraphics[scale=0.4]{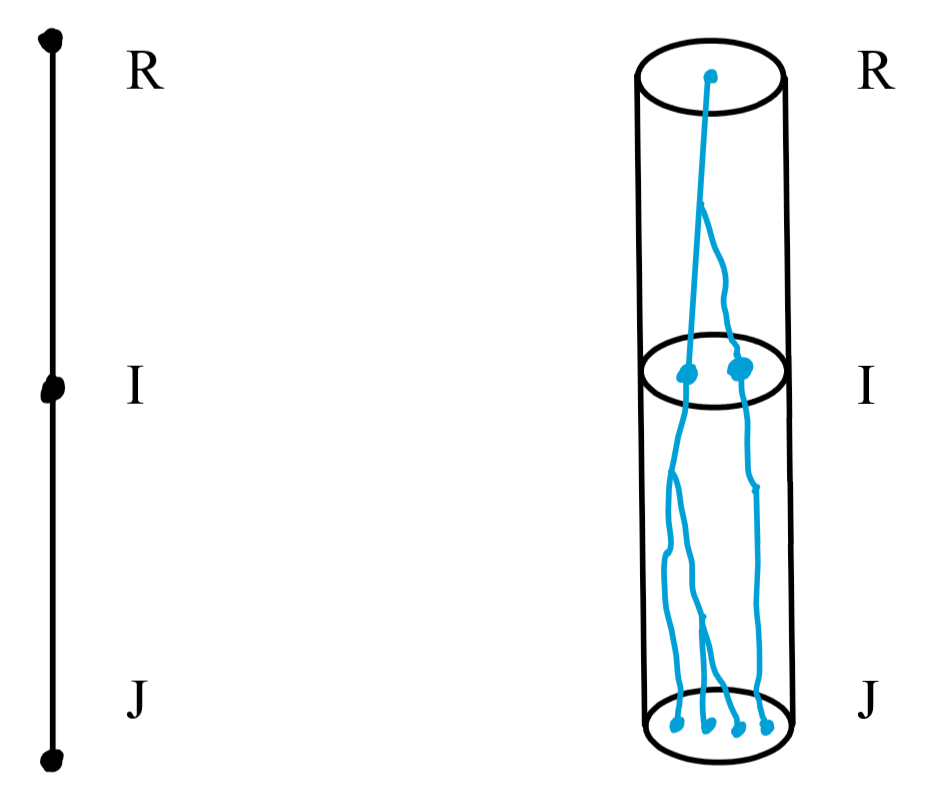}
    \caption{Left: The species tree with root $R$ and subsequent vertices $I$ and $J$.  Right:  A realization of the GDL process.  In the notation of Section 3, this example has $N_I = 2$ and $N_{I,1} = 3$ and $N_{I,2} = 1$.}
    \label{figure:GDL2node}
\end{figure}

}  

\section{Balanced tree on four species}

Throughout this section, the species tree $\sigma = (\mathcal{T},\mathbf{f})$ is assumed to have four leaf species $A,B,C,$ and $D$ and the topology has $A,B$ are siblings and $C,D$ are siblings.  \rev{A common way to represent this branching is through the Newick tree format; in this case $\mathcal{T}$ is equivalent to $((A,B),(C,D))$.  For four species, this branching pattern is \textit{balanced}.}

\rev{Next, we set out specific settings for the species tree for later reference.} The edge lengths are $s$ for stem edge, $f$ for the parent edge to $A$ and $B$, $g$ for the parent edge to $C$ and $D$, and $f_{n}', n \in \{A,B,C,D\}$ for the lengths of pendant edges incident to the leaves.  \rev{The most important species tree vertex to this analysis is the most recent common ancestor of the four species, which is labeled $I$.  Note that there is still a root edge and root vertex that is parent to $I$.  The root vertex is called $R$. The balanced tree topology with labelled internal vertices are depicted in the left frame of Figure \ref{figure:GDLspeciestrees}.  One realization of the GDL process applied to the balanced species tree is given in Figure \ref{figure:GDLuniformcopy}.}  
As suggested in Section 2, let $\PP'$ be the probability measure subject to the conditioning event where $N_A,N_B,N_C,$ and $N_D$ are positive.  Let $i_m, m \in \{a,b,c,d\}$ be the ancestor of sampled copy $m$ at $I$.  Conditioning further on $N_I$, the probability of each gene tree topology is written using $N_I$, using the probabilities $\rev{x = \PP'_{N_I}(i_a = i_b)},$ and $\rev{y = \PP'_{N_I}(i_c = i_d)}$.

\begin{figure}
    \centering
    \includegraphics[scale=0.5]{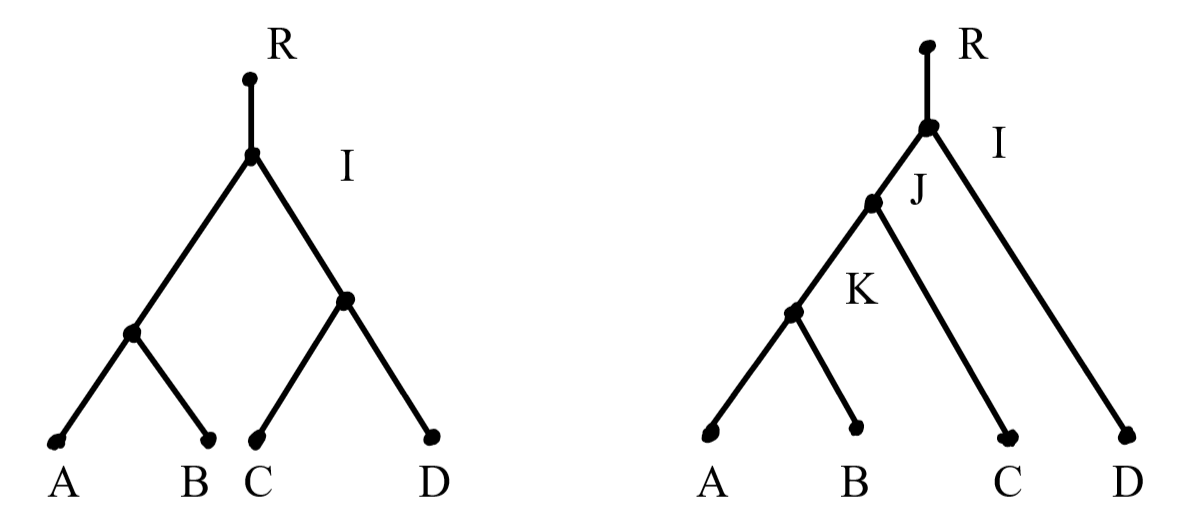}
    \caption{Left:  The balanced species tree with species $A,B,C,D$.  In Section 5, the length of the edge between $I$ and the parent of $A$ and $B$ is given by $f$, and the length of the edge between $I$ and the parent of $C$ and $D$ is given by $g$.  Right:  The caterpillar species with species $A,B,C,D$.  In Section 6, the length of the edge between $J$ and $K$ is given by $f$, and the length of the edge between $J$ and $I$ is given by $g$.}
    \label{figure:GDLspeciestrees}
\end{figure}

\begin{figure}
    \centering
    \includegraphics[scale=0.5]{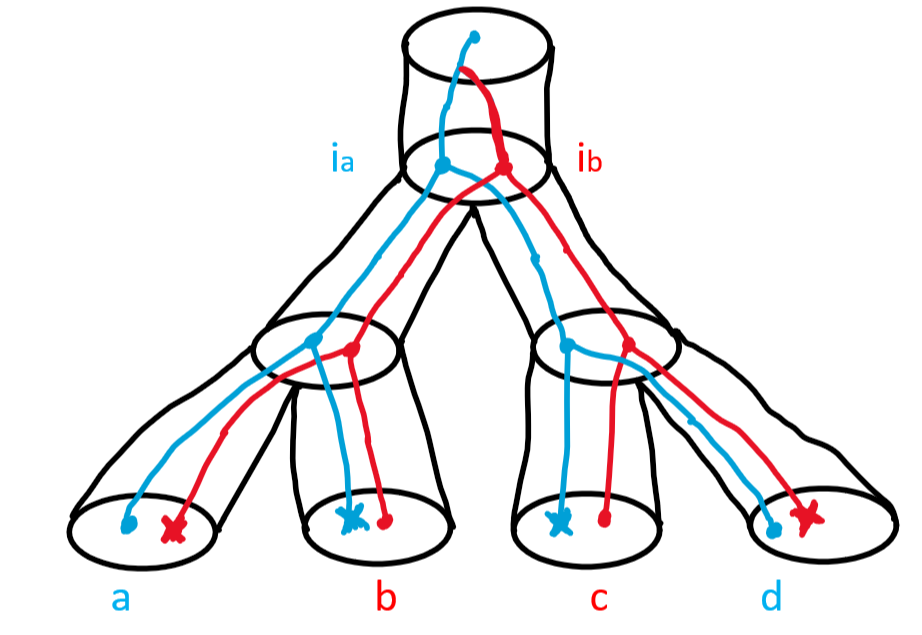}
    \caption{A realization of the GDL process on the balanced species tree $((A,B),(C,D))$ with ancestral copies.  The sampled copies are named $a,b,c,d$, respectively in the figure.  The left (blue) copy of $A$ was selected, so it is labeled $a$.  It descends from node $i_a$ at $I$.  Similarly, the right (red) copy of $B$ was selected, was labeled $b$, and descended from node $i_b$ at $I$.  The figure shows that $a$ and $d$ descend from the same individual at $I$, which is expressed as $i_a = i_d$.  Similarly, $i_b = i_c$ holds.}
    \label{figure:GDLuniformcopy}
\end{figure}

\rev{Here, we introduce notation defining the probability of observing a particular gene tree topology under uniform sampling.}  The values $h_w(x,y)$ each correspond to the probability of the $w$th gene tree topology as follows.  Let $w = 1$ correspond to the species tree topology.  Let $w = 2$ correspond to either of the two alternative balanced topologies, which have the same probability by exchangeability of sibling species. Let $w = 3$ correspond to any of the four caterpillar topologies in which the cherry is $\{a,b\}$.  Let $w = 4$ correspond to any of the caterpillar topologies in which the cherry is $\{c,d\}$.  Let $w = 5$ correspond to any of the remaining eight caterpillar topologies.

The main result of this section is that the species tree topology \rev{corresponds to the uniformly sampled gene tree with maximal probability.  The main implication of this result is that when more independent gene trees are given, the democratic vote estimator applied to the uniform sampled gene trees obtains the species tree topology with probability approaching one.}   Theorem \ref{thm:balanced} is \rev{a parallel result} to what is found in MSC.

\begin{theorem}
    \label{thm:balanced}
    Let $\sigma$ be a species tree on four species $A,B,C,D$ with the balanced topology.  Then $\rev{\mathbf{P}'(u(t) = \mathcal{T})}$ is maximized when $\mathcal{T}$ corresponds to the species tree topology.
\end{theorem}

To prove the result, we write explicit formulas for \rev{the values of $h_w(x,y)$.  Then, the result is equivalent to $h_1(x,y)$ being greater than $h_w(x,y)$ for $w \ne 1$.  Propositions \ref{prop:8}, \ref{prop:9}, and \ref{prop:10} contain more results than are necessary to prove the Theorem.  In particular, it is interesting to observe that $h_w(x,y) \geq h_5(x,y)$ for all $w$, even though it is apparent that $h_1(x,y) \geq h_5(x,y)$ without much analysis required.}

\rev{Before proving the Propositions, it will be useful to give an intuition behind the meaning of the results for phylogeneticists and other practitioners.  Before and after the uniform sampling step, the gene tree expresses meaningful signal about the overarching species tree.  Of course, if any species receives zero copies of a given gene, then the gene tree offers no information at all about that species.  We condition on survival of the gene in all extant species, so this does not occur.

However, signal is lost when $N_I$ is too large.  This is because when an individual from each extant species is chosen uniformly at random in ASTRAL-one, the individual's ancestors are chosen relative to the number of available choices.  The species tree signal is preserved only when the uniformly sampled copies from each species are chosen to have the most recent ancestry.  When $N_I$ is large, it is more likely that the sampled copies find common ancestry ``deeper'' in the tree, i.e. at a time before the speciation at vertex $I$.  As with the MSC, when four individuals in the same population have yet to ``coalesce'' going backward in time, they choose any of the three balanced topologies with probability $1/9$ and any of the twelve caterpillar (unbalanced) toplogies with probability $1/18$.  Proposition \ref{prop:6} provides the explicit formulas only.  Comparisons between $h_w(x,y)$ are shown in the later Propositions.  The fifteen possible uniformly sampled gene tree topologies with their associated $h_w$ numbers are given in Figure \ref{fig:GDLhtable}.}

\rev{Lastly, we introduce some new notation for this section and the next one only.  First, the notation $\mathbf{1}(E)$ is the indicator function of the event $E$ that may or may not occur for the gene tree.  For example $\mathbf{1}(N_I \geq 3)$ equals $1$ if $N_I \geq 3$ is true and equals $0$ if $N_I \geq 3$ is false.  The indicator function is used to give a single formula that applies for all cases of $N_I$.  Second, we have the commonly utilized binomial coefficient $$\mathcal{C}_{n,k} = \begin{pmatrix}
    n \\
    k
\end{pmatrix} = \frac{n!}{k!(n-k)!}, \ \ n! = n \times (n-1) \times (n-2) \times ... \times 2 \times 1.$$  For $n$ distinguishable objects, the number $\mathcal{C}_{n,k}$ corresponds to the number of distinct subsets of size $k$ that can be chosen from the set.}

\begin{prop}\label{prop:7}
(i) The balanced topologies have the following probabilities: \begin{align*}
    h_1(x,y) &= xy + \left(x(1-y) + (1-x)y \right) \frac{N_I - 2}{3N_I} \mathbf{1}(N_I \geq 3) \\
    &+ (1-x)(1-y) \frac{\begin{pmatrix}
    N_I \\
    2
    \end{pmatrix} - 2(N_I - 1) + 1}{9\begin{pmatrix}
    N_I \\
    2
    \end{pmatrix}} \mathbf{1}(N_I \geq 4). \\
    h_2(x,y) &= (1-x)(1-y) \bigg[ \bigg(\frac{1}{N_I (N_I - 1)} + \frac{N_I - 1}{3 \begin{pmatrix}
    N_I \\
    2
    \end{pmatrix}} \bigg)\mathbf{1}(N_I \geq 2) \\
    &+
    \frac{\begin{pmatrix}
    N_I \\
    2
    \end{pmatrix} - 2(N_I - 1) + 1}{9\begin{pmatrix}
    N_I \\
    2
    \end{pmatrix}}  \mathbf{1}(N_I \geq 4)  \bigg].
\end{align*} The caterpillar topologies have the following probabilities:
\begin{align*}
    h_3(x,y) &= x(1-y) \bigg[ \frac{1}{N_I} \mathbf{1}(N_I \geq 2) + \frac{N_I - 2}{3N_I} \mathbf{1}(N_I \geq 3)  \bigg] \\
    &+ (1-x)(1-y) \frac{\begin{pmatrix}
    N_I \\
    2
    \end{pmatrix} - 2(N_I - 1) + 1}{18\begin{pmatrix}
    N_I \\
    2
    \end{pmatrix}} \mathbf{1}(N_I \geq 4)  \\
    h_4(x,y) &= (1-x)y \bigg[ \frac{1}{N_I} \mathbf{1}(N_I \geq 2) + \frac{N_I - 2}{3N_I} \mathbf{1}(N_I \geq 3)  \bigg] \\
    &+ (1-x)(1-y) \frac{\begin{pmatrix}
    N_I \\
    2
    \end{pmatrix} - 2(N_I - 1) + 1}{18\begin{pmatrix}
    N_I \\
    2
    \end{pmatrix}}  \mathbf{1}(N_I \geq 4)  \\
    h_5(x,y) &= (1-x)(1-y) \left( \frac{N_I-2}{3N_I(N_I-1)} \mathbf{1}(N_I \geq 3) + \frac{\begin{pmatrix}
    N_I \\
    2
    \end{pmatrix} - 2(N_I - 1) + 1}{18\begin{pmatrix}
    N_I \\
    2
    \end{pmatrix}} \mathbf{1}(N_I \geq 4) \right).
\end{align*}
\end{prop}

\begin{proof}
The proof is a basic application of the law of total probability and counting  \cite{legriedGDL2021}, \cite{hill2022species}.  \rev{We make one note about the number of favorable choices among the total $\mathcal{C}_{N_I,2}$.  Provided there are at least four choices of individuals at $I$ and that $i_a \ne i_b$ and $i_c \ne i_d$ in the calculation of $h_1(x,y)$, exchangeability implies $i_a$ and $i_b$ is equally likely to be any of the $\mathcal{C}_{N_I,2}$ ancestor choices to occupy.  Independently, the pair $\{i_c,i_d\}$ must choose a disjoint subset from $\{i_a,i_b\}$ in order to obtain the $w = 1$ case.  There are $N_I-1$ ways to choose a subset of the form $\{i_a,\tau\}$ where $\tau \ne i_a$ and $N_1-1$ ways to choose a subset of the form $\{i_b,\tau\}$ where $\tau \ne i_b$, but counting both results in double-counting the subset $\{i_a,i_b\}$.  So there are $2(N_I - 1) - 1$ unfavorable subsets for obtaining $w = 1$.  A similar choice is made to set up the possibility of any other realization of $w$.}   
\end{proof}

The next Propositions provide a ranking of probabilities that hold for any $N_I$ and when $x,y \geq 1/N_I$.  \rev{In Proposition \ref{prop:8}, showing $h_j(x,y) \geq h_{j'}(x,y)$ amounts to showing that $h_j - h_{j'}$ is non-negative for all choices of $x$ and $y$ that satisfy $1/N_I$.  Note that $x = \mathbf{P}'_{N_I}(i_a = i_b)$ and $y = \mathbf{P}'_{N_I}(i_c = i_d)$ have conditioning on the observed value of $N_I$.  In Proposition \ref{prop:8}, the criteria that $x,y$ are both at least $1/N_I$ does not make sense for the measure $\mathbf{P}'$, as $N_I$ is still random.  Showing $h_1 \geq h_2$ requires minimization of a linear function as in \cite{legriedGDL2021}, but $h_2 \geq h_5$ is more apparent by simply comparing like terms.}

\begin{prop}\label{prop:8}
Suppose $x,y \geq 1/N_I$.  Then $h_1(x,y) \geq h_2(x,y) \geq h_5(x,y)$.
\end{prop} 

\begin{proof}We start by showing $h_1 \geq h_2$.  The difference is \begin{align*}h_1(x,y) - h_2(x,y) &= xy + \left(x(1-y) + (1-x)y \right) \frac{N_I-2}{3N_I}\mathbf{1}(N_I \geq 3) \\
&- (1-x)(1-y) \frac{2N_I-1}{3N_I(N_I-1)}\mathbf{1}(N_I \geq 2).\end{align*} When taking the partial derivative in $x$, we split $1-2y$ into $1-y-y$.  Then the partial derivative in $x$ is \begin{align*}&y + (1-2y) \frac{N_I-2}{3N_I}\mathbf{1}(N_I \geq 3) + (1-y) \frac{2N_I - 1}{3N_I(N_I-1)}\mathbf{1}(N_I \geq 2) \\
&= y \left(1 - \frac{N_I-2}{3N_I}\mathbf{1}(N_I \geq 3)\right) + (1-y) \left( \frac{N_I - 2}{3N_I}\mathbf{1}(N_I \geq 3) + \frac{2N_I-1}{3N_I(N_I - 1)}\mathbf{1}(N_I \geq 2) \right)
\end{align*} The same holds when we compute the partial derivative in $y$.  The function $h_1 - h_2$ is linear in both $x$ and $y$, so the minimum value of $h_1 - h_2$ is obtained by substituting $x = y = 1/N_I$.  Then $$h_1(x,y) - h_2(x,y) \geq h_1(1/N_I,1/N_I) - h_2(1/N_I,1/N_I) = \begin{cases}
1/N_I^3 & N_I \geq 3 \\
1/N_I^2 & N_I = 1,2.
\end{cases}$$ So $h_1 \geq h_2$.

To show $h_2 \geq h_5$, we only need to subtract.  The difference is \begin{align*}
    &h_2(x,y) - h_5(x,y) \\
    &= (1-x)(1-y) \bigg( \frac{5}{6} + \left( \frac{1}{N_I(N_I-1)} + \frac{2}{3N_I} - \frac{N_I-2}{3N_I(N_I - 1)} \right) \mathbf{1}(N_I \geq 3) \\
    &+ \frac{\begin{pmatrix}
    N_I \\
    2
    \end{pmatrix} - 2(N_I - 1) + 1}{18\begin{pmatrix}
    N_I \\
    2
    \end{pmatrix}} \mathbf{1}(N_I \geq 4) \bigg) \\
    &= (1-x)(1-y) \bigg(\frac{5}{6} + \frac{N_I+3}{3N_I(N_I-1)} \mathbf{1}(N_I \geq 3) \\
    &+ \frac{\begin{pmatrix}
    N_I \\
    2
    \end{pmatrix} - 2(N_I - 1) + 1}{18\begin{pmatrix}
    N_I \\
    2
    \end{pmatrix}} \mathbf{1}(N_I \geq 4) \bigg),
\end{align*} which is clearly non-negative for all choices of $x$ and $y$.
\end{proof}

\rev{We now proceed to comparisons with $h_3(x,y)$ and $h_4(x,y)$.  By exchangeability, switching the roles of $A$ and $B$ yields the same probability distribution of gene trees, and the results contained in Propositions \ref{prop:9} and \ref{prop:10} differ only in this exchange.  So, we prove Proposition \ref{prop:9} only, with the understanding that the proof of Proposition \ref{prop:10} is largely identical.  Linear optimization is not required in these proofs -- only a (delicate) comparison of like terms is needed.}

\begin{prop}\label{prop:9} Suppose $x,y \geq 1/N_I$.  Then $h_1(x,y) \geq h_3(x,y) \geq h_5(x,y)$.
\end{prop}

\begin{proof}
We start by showing that $h_1 \geq h_3$.  The difference is \begin{align*}
    h_1(x,y) - h_3(x,y) &= xy + (1-x)y \frac{N_I - 2}{3N_I}\mathbf{1}(N_I \geq 3) \\
    &- x(1-y) \frac{1}{N_I}\mathbf{1}(N_I \geq 2) \\
    &+ (1-x)(1-y) \frac{\begin{pmatrix}
    N_I \\
    2
    \end{pmatrix} - 2(N_I - 1) + 1}{18 \begin{pmatrix}N_I \\
    2
    \end{pmatrix}} \mathbf{1}(N_I \geq 4).
\end{align*} The only negative part is in the second line, and it can be split into $$-x(1-y) \frac{1}{N_I} \mathbf{1}(N_I \geq 2) = -x\frac{1}{N_I} \mathbf{1}(N_I \geq 2) + xy \frac{1}{N_I} \mathbf{1}(N_I \geq 2).$$ However, $y \geq 1/N_I$ and combining this negative term with the positive $xy$ term in the original expression implies $$xy  -x\frac{1}{N_I} \mathbf{1}(N_I \geq 2) \geq 0.$$ This dispenses with all negative terms, showing that $h_1 \geq h_3$.

Now we show that $h_3 \geq h_5$.  The difference is \begin{align*}
    h_3(x,y) - h_5(x,y) &= x(1-y) \bigg[ \frac{1}{N_I} \mathbf{1}(N_I \geq 2) + \frac{N_I - 2}{3N_I} \mathbf{1}(N_I \geq 3)  \bigg] \\
    &-(1-x)(1-y) \frac{N_I - 2}{3N_I(N_I - 1)}\mathbf{1}(N_I \geq 3).
\end{align*} Using $x \geq 1/N_I$ implies \begin{align*}
    &x(1-y) \bigg[ \frac{1}{N_I} \mathbf{1}(N_I \geq 2) + \frac{N_I - 2}{3N_I} \mathbf{1}(N_I \geq 3)  \bigg] - (1-y) \frac{N_I - 2}{3N_I(N_I - 1)}\mathbf{1}(N_I \geq 3) \bigg] \\
    &\geq (1-y) \bigg[ \frac{1}{2^2} + \frac{2N_I - 1}{3N_I(N_I-1)}\mathbf{1}(N_I \geq 3) \bigg] \geq 0.
\end{align*} We have covered all negative components of $h_3 - h_5$ with corresponding positive numbers, so we conclude $h_3 \geq h_5$.
\end{proof}

\begin{prop}\label{prop:10} Suppose $x,y \geq 1/N_I$.  Then $h_1(x,y) \geq h_4(x,y) \geq h_5(x,y)$.
\end{prop}

\begin{proof}
The proof is identical to that of Proposition \ref{prop:9}, except we switch the roles of $x$ and $y$.
\end{proof}

\begin{figure}
    \includegraphics[scale=0.5]{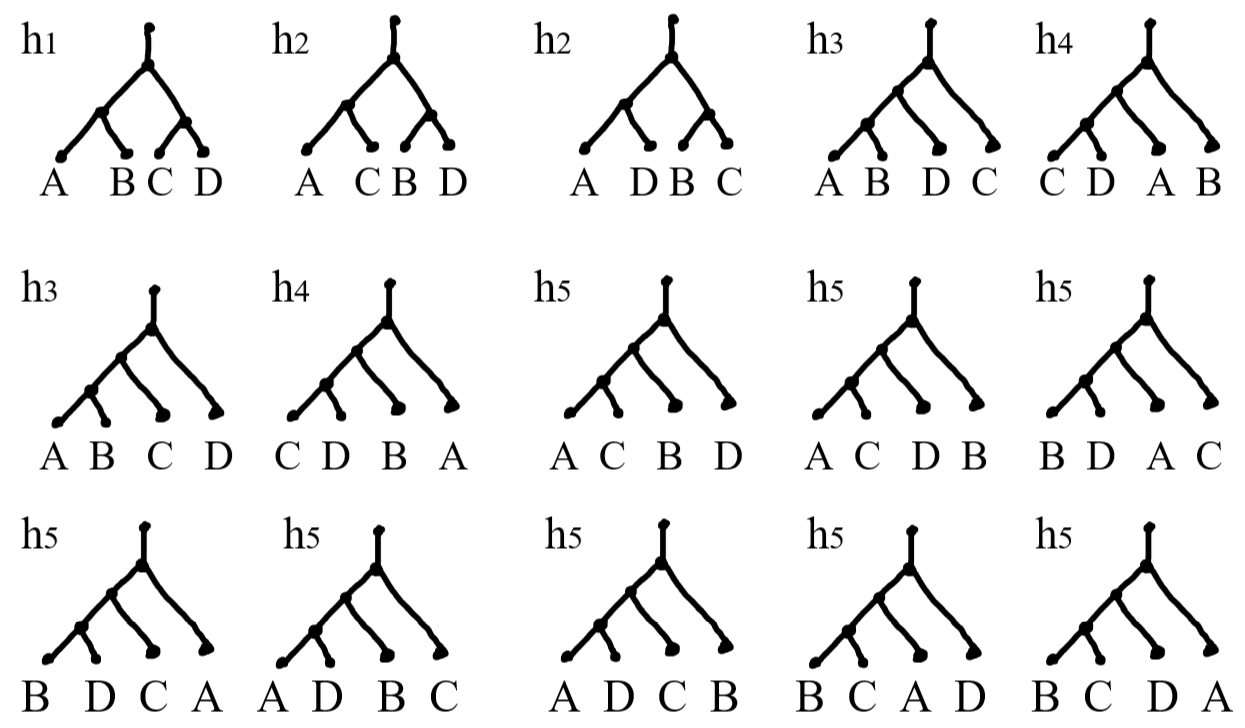}
    \caption{Labelling of the possible uniformly sampled gene trees by probability in the case of the balanced species tree.}
    \label{fig:GDLhtable}
\end{figure}

\section{Caterpillar topology on four leaves.}

In this section, the species tree $\sigma$ is assumed to have four leaf species $A,B,C,$ and $D$ and the \rev{branching pattern} has $A,B$ are siblings, followed by $C$ as an outgroup, and followed by $D$ as a further outgroup.  The Newick tree format is $(((A,B),C),D)$.  \rev{This branching pattern is \textit{unbalanced} or \textit{caterpillar}.}

The edge lengths are $s$ for the stem edge, $f$ for the parent edge to $A$ and $B$, $g$ for the parent edge to the most recent common ancestor of $A,B,$ and $C$, and $f_n', n \in \{A,B,C,D\}$ for the lengths of the pendant edges incident to the leaves.  The root vertex is labeled $I$, as in the previous sections.  Let $\PP'$ be the probability measure subject to the conditioning event where $N_A,N_B,N_C,$ and $N_D$ are positive.  Let $I$ be the most recent common ancestor to the four species; $J$ be the most recent common ancestor to $A,B,C$; and $K$ be the most recent common ancestor to $A$ and $B$.  \rev{The caterpillar topology with these labelled vertices is given in the right frame of Figure \ref{figure:GDLspeciestrees}.} The names of the ancestors of sampled copies are $k_{\ell}, \ell \in \{a,b\}$; $j_{\ell'}, \ell' \in \{a,b,c\}$; and $i_{\ell''}, \ell'' \in \{a,b,c,d\}$.  

Conditioning further on $N_I$, the probability of each gene tree topology is written using $N_I$, with the measure $\PP'_{N_I}$. \rev{Relevant conditional probabilities are \begin{align*}
w &= \mathbf{P}'_{N_I}(j_a = j_b) \\
x &= \mathbf{P}'_{N_I}(i_a = i_b = i_c|j_a = j_b) \\
y &= \mathbf{P}'_{N_I}(i_a = i_b = i_c, |\{j_a,j_b,j_c\}| = 3| j_a \ne j_b) \\
v &= \mathbf{P}'_{N_I}(i_a = i_b = i_c, j_a = j_c|j_a \ne j_b) = \mathbf{P}'_{N_I}(i_a = i_b = i_c,j_b = j_c|j_a \ne j_b) \\
z &= \mathbf{P}'_{N_I}(i_a = i_b \ne i_c|j_a \ne j_b) = \mathbf{P}'_{N_I}(i_a = i_c \ne i_b|j_a \ne j_b) = \mathbf{P}'_{N_I}(i_b = i_c \ne i_a|j_a \ne j_b).
\end{align*} We must isolate these particular quantities because there are more options for the ancestry of the sampled copies from $A$ and $B$.  With an intuition based in coalescent theory, we proceed from the leaves and work backwards.  We first check whether the copies from $A$ and $B$ find ancestry at $J$, which is given by $w$.  Conditioned on $j_a = j_b$, we no longer need to consider how the copies merged there -- all that matters is how this individual merges with the sampled copy from $C$.  In computing $x$, it is already known that $i_a = i_b$, so the random content is whether $i_a = i_c$.  On the other hand, if $j_a \ne j_b$, then there are three equivalent ancestral copies to those of $A,B,C$ in computing $y, v, $ and $z$.   

}

\rev{As in the previous Section, we provide notation for each possible uniformly sampled gene tree with a visual aid in Figure \ref{fig:GDLktable}.}  Throughout, we let $k_{w'}$ denote the probability of gene tree $t_{w'}$, with $w'$ indexed as follows. The species tree is listed first, followed by the possible anomalous gene trees.  Let $k_1$ denote the probability of caterpillar tree $(((a,b),c),d)$, let $k_2$ denote the probability of the balanced tree $((a,b),(c,d)$, and let $k_3$ denote the probability of each alternate balanced tree $((a,c),(b,d))$ and $((a,d),(b,c))$.  Then $k_{w'}, w' \in \{4,5,...,8\}$ represent the other caterpillar topologies.

The main theorem of this section establishes possible anomaly zones for the caterpillar tree.  \rev{That is, given any choice of birth rate $\lambda$ and death rate $\mu$, the caterpillar species tree with some choices of branch lengths produces a gene tree distribution where the uniformly sampled gene tree can have an alternate topology have maximal probability.  We show that such anomalous gene trees can correspond to only balanced quartets.  The species tree topology must have probability greater than any other caterpillar topology, for any choice of branch lengths.  Theorem \ref{thm:caterpillar} is also a parallel result to what is found in MSC.}

\begin{theorem}
    Let $\sigma$ be a species tree on four species $A,B,C,D$ with the caterpillar topology.  Then $\rev{\mathbf{P}'(u(t) = \mathcal{T})}$ is maximized when $\mathcal{T}$ corresponds to either the species tree or one of the balanced rooted quartets.
\end{theorem}

\rev{As with the previous Theorem, the proof will proceed by writing explicit formulas for $k_{w'}(w,x,y,v,z)) = k_{w'}$, and show that $k_1 \geq k_{w'}$ for every $w' \in \{4,5,6,7,8\}$.  This is all that is necessary to prove Theorem \ref{thm:caterpillar}.  In this paper, it will not be possible to fully characterize the relationship between $k_1,k_2,$ and $k_3$.  Instead, we introduce a new framework for understanding anomaly zones under GDL and prove partial results in the next Section. It will be shown that $t_2$ has greater probability than $t_3$, implying three possible rankings of probability.  They are (I) $k_1 > k_2 > k_3$ (no anomalous gene trees); (II) $k_2 > k_1 > k_3$ ($t_2$ is anomalous); and (III) $k_2 > k_3 > k_1$ (all balanced trees are anomalous).  Any subsequent analysis of uniformly sampled GDL trees should consider whether a given caterpillar species tree belongs to (I), (II), or (III).}

\rev{The first result is a technical Lemma that establishes sufficient bounds on $w$ and $x$.  By stating that $w$ or $x$ is at least $1/N_I$, one is implicitly stating that the ancestors of sampled copies of $A$ and $B$ are ``positively correlated'' in the sense that they are more likely to choose the same ancestor relative to a random uniform sampling of individuals at $I$.}

\begin{lemma} \label{lem:1}
Conditioned on $N_I$, we have $w,x \geq \frac{1}{N_I}$.
\end{lemma}

\begin{proof} Both bounds follow mostly from Lemma 1 of \cite{legriedGDL2021}.
The bound on $w$ follows by conditioning on $N_K$, the total number of individuals surviving to $K$.  Letting $(M_k)_{k=1}^{N_I}$ be the descendants of each individual at $I$ surviving to $K$, the selection of $j_a$ and $j_b$ are independent.  We have \begin{align*}
w = \PP'_{N_K}(j_a = j_b) &= \EE'_{N_I}\left[ \frac{\sum_{k=1}^{N_I}M_k^2}{(\sum_{k=1}^{N_I}M_k)^2} \bigg| (M_k)_{k=1}^{N_I}\right] \\
&\geq \EE'_{N_I} \left[\frac{1}{N_I}\right] = \frac{1}{N_I}.
\end{align*} The inequality utilized is the quadratic mean inequality.

Now we consider the bound on $x$.  Conditioned on $j_a = j_b$, the event $i_a = i_b$ is known to occur.  So $x = \PP'_{N_I}(i_a = i_c|j_a = j_b)$.  Letting $(M_j)_{j=1}^{N_I}$ be the descendants of each individual at $I$ surviving to $J$, the selection of $j_a$ and $j_c$ are proportional to these weights and are independent.  Then \begin{align*}
\PP'_{N_I}(i_a = i_c|j_a = j_b)&= \EE'_{N_I} \left[ \frac{\sum_{j=1}^{N_I}M_j^2}{\left(\sum_{j=1}^{N_I}M_j\right)^2} \bigg| (M_j)_{j=1}^{N_I}\right],
\end{align*} and we obtain the lower bound of $1/N_I$.
\end{proof}

\rev{The next Lemma follows similarly by considering the case where $i_a = i_b = i_c$.  Knowing that the copies of $A,B,$ and $C$ find a common ancestor at $I$ rather than sometime between $I$ and $R$ is irrelevant to asking about the ancestry of the copies of $A$ and $B$ at $J$.

\begin{lemma} \label{lem:2}
    Conditioned on $N_I$, we have $wx \geq (1-w)v$.
\end{lemma}

\begin{proof}
    We have \begin{align*}
        wx &= \mathbf{P}'_{N_I}(i_a = i_b = i_c, j_a = j_b) \\
        (1-w)v &= \mathbf{P}'_{N_I}(i_a = i_b = i_c, j_a = j_c \ne j_c).
    \end{align*} Then \begin{align*}&wx - (1-w)v \\
    &= \left(\mathbf{P}'_{N_I}(j_a = j_b|i_a = i_b = i_c) - \mathbf{P}'_{N_I}(j_a = j_c \ne j_b|i_a = i_b = i_c)\right)\mathbf{P}'_{N_I}(i_a = i_b = i_c).\end{align*} The factor contained in the parentheses is shown to be non-negative through a similar method to Lemma \ref{lem:1}.
\end{proof}

}

\rev{The next several Propositions are stated separately because }

\begin{prop} \label{prop:11}
The caterpillar topology $(((a,b),c),d)$ has the probability \begin{align*}
k_1 &= wx + w(1 - x) \left(1 - \frac{2}{N_I}\right) \frac{1}{3}\mathbf{1}(N_I \geq 3) \\
&+ (1-w)y \frac{1}{3} + (1-w) z \left(1 - \frac{2}{N_I}\right) \frac{1}{3}\mathbf{1}(N_I \geq 3) \\
&+ (1-w)(1 - y \rev{-2v} - 3z)\left(1 - \frac{3}{N_I}\right)\frac{1}{18}\mathbf{1}(N_I \geq 4).
\end{align*}
\end{prop}

\begin{proof}
This follows from the law of total probability. \rev{Recall} that when $i_a,i_b,i_c,i_d$ are distinct, the probability of the correct caterpillar topology is $1/18$.
\end{proof}

Next, let $k_2$ be the probability of the balanced topology $((a,b),(c,d))$ and $k_3$ be the probability of the alternate balanced topology $((a,c),(b,d))$.  Exchangeability of $a$ and $b$ means that $((a,d),(b,c))$ also has probability $k_3$.

\begin{prop} \label{prop:12}
The balanced topology $((a,b),(c,d))$ has probability \begin{align*}
k_2 &= w(1-x) \left[ \frac{1}{N_I} \mathbf{1}(N_I \geq 2) + \left(1 - \frac{2}{N_I}\right) \frac{1}{3}\mathbf{1}(N_I \geq 3)\right] \\
&+ (1-w)z \left[ \frac{1}{N_I} \mathbf{1}(N_I \geq 2) + \left(1 - \frac{2}{N_I}\right) \frac{1}{3}\mathbf{1}(N_I \geq 3)\right] \\
&+ (1-w)(1- y \rev{-2v} - 3z)\left[ \frac{1}{3}\frac{1}{N_I} \mathbf{1}(N_I \geq 3) + \left(1 - \frac{3}{N_I}\right) \frac{1}{9}\mathbf{1}(N_I \geq 4)\right].
\end{align*} The balanced topology $((a,c),(b,d))$ has  probability \begin{align*}
k_3 &= (1-w)z \left[ \frac{1}{N_I} \mathbf{1}(N_I \geq 2) + \left(1 - \frac{2}{N_I}\right) \frac{1}{3}\mathbf{1}(N_I \geq 3)\right] \\
&+(1-w)(1- y \rev{-2v}- 3z)\left[ \frac{1}{3}\frac{1}{N_I} \mathbf{1}(N_I \geq 3) + \left(1 - \frac{3}{N_I}\right) \frac{1}{9}\mathbf{1}(N_I \geq 4)\right].
\end{align*}
\end{prop}

Now, we consider caterpillar topologies where the unrooted topology is the same as the species tree.  Let $k_4$ be the probability of $(((a,b),d),c)$ and $k_5$ be the probability of $(((c,d),a),b)$.  The topology $(((c,d),b),a)$ has the same probability $k_5$.

\begin{prop} \label{prop:13}
The caterpillar topology $(((a,b),d),c)$ has probability \begin{align*}
k_4 &= w(1-x) \left[ \frac{1}{N_I} \mathbf{1}(N_I \geq 2) + \left(1 - \frac{2}{N_I}\right)\frac{1}{3}\mathbf{1}(N_I \geq 3)\right] \\
&+ (1-w)z \left[ \frac{1}{N_I}\mathbf{1}(N_I \geq 2) + \left(1 - \frac{2}{N_I}\right) \frac{1}{3}\mathbf{1}(N_I \geq 3)\right] \\
&+ (1-w)(1-y - 3z) \left(1 - \frac{3}{N_I}\right) \frac{1}{18}\mathbf{1}(N_I \geq 4).
\end{align*} The caterpillar topology $(((c,d),a),b)$ has probability \begin{align*}
k_5 &= (1-w)(1-y \rev{-2v}-3z) \left[ \frac{1}{3} \frac{1}{N_I}\mathbf{1}(N_I \geq 3) + \left(1 - \frac{3}{N_I}\right) \frac{1}{18} \mathbf{1}(N_I \geq 4)\right].
\end{align*}
\end{prop}

Lastly, we consider four caterpillar topologies where the unrooted topology is $((a,c),(b,d))$.  Let $k_6$ be the probability of $(((a,c),b),d)$, $k_7$ be the probability of $(((a,c),d),b)$, $k_8$ be the probability of $(((b,d),a),c)$, and $k_9$ be the probability of $(((b,d),c),a)$.  The other caterpillar topologies with unrooted topology $((a,d),(b,c))$ are represented in this calculation.

\begin{prop}\label{prop:14}
The caterpillar topology $(((a,c),b),d)$ has probability \begin{align*}
k_6 &= \rev{(1-w)v} + (1-w)y \frac{1}{3} + (1-w)z \left(1 - \frac{2}{N_I}\right) \frac{1}{3}\mathbf{1}(N_I \geq 3) \\
&+(1-w)(1-y\rev{-2v}-3z) \left(1 - \frac{3}{N_I}\right)\frac{1}{18}\mathbf{1}(N_I \geq 4).
\end{align*} The caterpillar topology $(((a,c),d),b)$ has probability \begin{align*}
k_7 &= (1-w)z \left[ \frac{1}{N_I} \mathbf{1}(N_I \geq 2) + \left(1 - \frac{2}{N_I}\right)\frac{1}{3}\mathbf{1}(N_I \geq 3)\right] \\
&+ (1-w)(1-y\rev{-2v}-3z) \left[\frac{1}{3}\frac{1}{N_I}\mathbf{1}(N_I \geq 3) + \left(1 - \frac{3}{N_I}\right) \frac{1}{18}\mathbf{1}(N_I \geq 4)\right].
\end{align*} The caterpillar topologies $(((b,d),a),c)$ and $(((b,d),c),a)$ have probability \begin{align*}
k_8 &= (1-w)(1-y\rev{-2v}-3z) \left[\frac{1}{3}\frac{1}{N_I}\mathbf{1}(N_I \geq 3) + \left(1 - \frac{3}{N_I}\right) \frac{1}{18}\mathbf{1}(N_I \geq 4)\right].
\end{align*}
\end{prop}

We now rank these probabilities where possible and determine when the average $k_1$ is not the maximal probability.  \rev{These results follow from the two Lemmas and a similar comparison of like terms as in the previous Section.}

\begin{prop} \label{prop:15} We have $k_1 \geq k_4, k_5, k_6, k_7, k_8$.
\end{prop}

\begin{proof}
For $k_1 \geq k_4$, we eliminate the non-corresponding terms.  This leaves \begin{align*}
k_1 - k_4 &= wx + (1-w)y \frac{1}{3} - \left[w(1-x) + (1-w)z\right]\frac{1}{N_I}\mathbf{1}(N_I \geq 2).
\end{align*} Because $x \geq \frac{1}{N_I}$, we have \begin{align*}wx - w(1-x) \frac{1}{N_I} \mathbf{1}(N_I \geq 2) &= w \left[ x\left(1 + \frac{1}{N_I}\mathbf{1}(N_I \geq 2)\right) - \frac{1}{N_I} \mathbf{1}(N_I \geq 2) \right] \\
&\geq \frac{w}{N_I}\end{align*} Because $w \geq \frac{1}{N_I}$, we have \begin{align*}
\frac{w}{N_I} - (1-w)z \frac{1}{N_I}\mathbf{1}(N_I \geq 2) &= w \left(\frac{1}{N_I} + z\frac{1}{N_I}\mathbf{1}(N_I \geq 2)\right) - z \frac{1}{N_I}\mathbf{1}{N_I \geq 2} \\
&\geq \frac{1}{N_I^2}.
\end{align*} This covers all the negative terms, so $k_1 \geq k_4$, as needed.

We now show $k_1 \geq k_5$.  We first observe that $$wx + w(1-x)\left(1 - \frac{2}{N_I}\right) \frac{1}{3}\mathbf{1}(N_I \geq 3) = \frac{w}{3}\left(1 - \frac{2}{3N_I} + \frac{2x}{3}\right) \geq \frac{1}{3N_I}.$$ This is at least $$(1-w)(1-y\rev{-2v}-3z) \frac{1}{3N_I},$$ which is sufficient to conclude $k_1 \geq k_5$ by cancelling other terms of $k_5$.

\rev{That $k_1 \geq k_6$ follows directly from Lemma \ref{lem:2}.} The remaining claims follow by cancelling terms and using the inequality from the previous paragraph.
\end{proof}

We also have that $k_2 \geq k_3$, which is apparent.  This establishes the trichotomy stated in Theorem 3.

\begin{prop} \label{prop:16}
We have $k_2 \geq k_3$.
\end{prop}

\begin{figure}
    \includegraphics[scale=0.5]{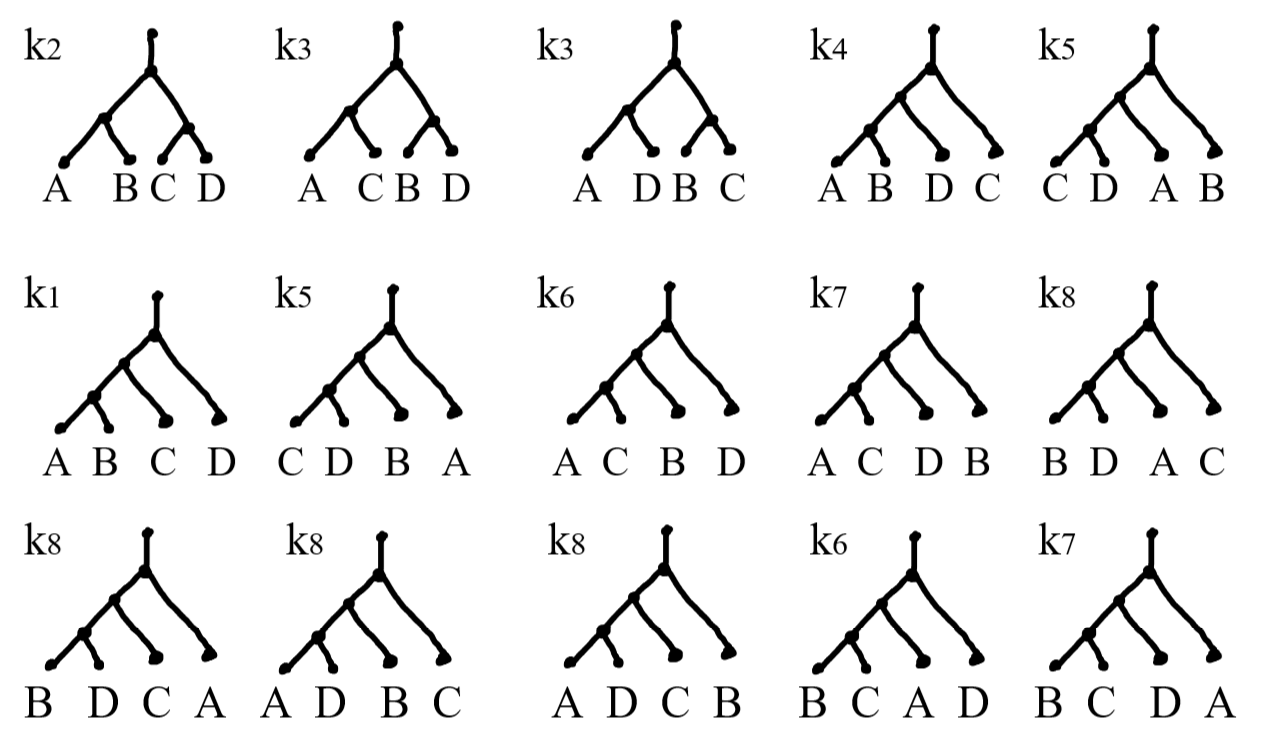}
    \caption{Labelling of the possible uniformly sampled gene trees by probability in the case of the caterpillar species tree.}
    \label{fig:GDLktable}
\end{figure}

\rev{\subsection{Existence of anomaly zones}

In this subsection, we consider some cases where anomaly zones arise.  Under MSC, this amounts to choosing a caterpillar species tree with small enough branch lengths.  In this paper, we consider branch lengths as well as the birth rate $\lambda$.  Theorem \ref{thm:caterpillar} alone says that there may be some branch length settings that put the species tree in the anomaly zone, but it does not actually prove that the anomaly zone exists.

In this Section, we introduce a limited setting where anomaly zones exist, namely when the interior branch lengths are vanishingly small.  One might expect that if the branch lengths are larger that there may be a particular setting of birth and death rates where anomaly zones exist.  In either case, the anomaly zone seems quite remote, as the birth rate needs to be fairly large.  Proving a negative result requires a different mode of analysis than used in the previous two Sections.  First, the probabilities $w,x,y,v,z$ must be computed exactly in the regime we consider.  Second, because they are computed exactly, their dependencies on $N_I$ must also be not so extreme to compute probabilities under the measure $\mathbf{P}'$ \textit{without} conditioning on $N_I$.  Because $N_I$ is a modified geometric random variable, the challenging infinite sums we consider are indeed summable.

Next, we introduce notation needed to prove the technical results of this Section.  Let $\eta_I$ be the length of the interior child edge of $I$ and $\eta_J$ be the length of the interior child edge of $J$.  Provided $\eta_I$ and $\eta_J$ are sufficiently small, there is an anomaly zone for some choices of $\lambda$ and $\mu$.  The explicit formulas are described in this section, and we provide figures to display the results for some chosen values of $\mu$ in Section 5.2.

\begin{theorem} \label{thm:caterpillar}
    Let $\sigma = (\mathcal{T},\mathbf{f})$ be a species tree on four species $A,B,C,D$ with the caterpillar topology $\mathcal{T}$ and let $\mu$ be fixed.  Let $\mathcal{T}'$ be any choice of balanced topology joining the uniformly sampled copies $a,b,c,d$.  Then as $\eta_I$ and $\eta_J$ converge to $0$, there exists a positive number $\Lambda$ such that $\mathbf{P}'(u(t) = \mathcal{T}') > \mathbf{P}'(u(t) = \mathcal{T})$ for all birth rate choices $\lambda \geq \Lambda$.
\end{theorem}

Before establishing the Theorem, we first prove the convergence of $w,x,y,z$ in these limits.  Proposition \ref{prop:17} uses a technical probability result (the Portmanteau Theorem) to show that $w$ and $x$ essentially converge to a number indicating the underlying ancestral copies are ``uncorrelated'' in the uniform sample.  Similarly, $v,y,z$ converge to $0$ as the type of transition described is not permitted on a branch undergoing no evolution.

\begin{prop} \label{prop:17}
    Let $N_I$ and $w,x,y,z$ be defined as before.  Then as $\eta_I, \eta_J \rightarrow 0$, we have \begin{align*}
        w,x \rightarrow \frac{1}{N_I} \ \textnormal{and} \ (1-w)v,y,z \rightarrow 0.
    \end{align*}
\end{prop}

\begin{proof}
    As $\eta_I \rightarrow 0$, we have the offspring numbers $(M_j)_{j=1}^{N_I}$ each converge to $1$ in distribution as $\eta_I \rightarrow 0$.  Because the $M_j$ are independent, the vector $(M_1,...,M_I)$ converges in distribution to the vector $(1,...,1)$ as $\eta_I \rightarrow 0$.  The function $\mathcal{F}(y_1,...,y_{N_I}) = \sum_{j=1}^{N_I}y_j^2/(\sum_{j=1}^{N_I}y_j)^2$ is bounded and continuous around $(1,...,1)$, so the Portmanteau Theorem implies that $$\mathbf{E}'_{N_I}[\mathcal{F}(M_1(\eta_I),...,M_I(\eta_I)) \rightarrow \mathbf{E}'_{N_I}[\mathcal{F}(M_1(0),...,M_{N_I}(0))] = \frac{1}{N_I}$$ as $\eta_I \rightarrow 0$.  For the connections in probability, see \cite{durrett_2010}.

    Conditioned on $j_a = j_b$, the event $i_a = i_c$ is equivalent to the coalescence of a single independent pair, i.e. $x \rightarrow 1/N_I$ as $\eta_I \rightarrow 0$.  A similar conclusion holds for $w \rightarrow 1/N_I$ as both $\eta_I, \eta_J$ tend to $0$.  In the event that $j_a \ne j_b$, it is not possible for $i_a = i_b$ as $\eta_I \rightarrow 0$, so $(1-w)v,y,z \rightarrow 0$.
\end{proof}

Next, we consider the differences between $k_1,k_2,k_3$ in the case where $\eta_I$ and $\eta_J$ are in the limiting case of $0$.  We have \begin{align*}
    k_1 &\rightarrow \frac{1}{N_I^2} + \frac{1}{N_I}\left(1 - \frac{1}{N_I}\right)\left(1 - \frac{2}{N_I}\right) \frac{1}{3}\mathbf{1}(N_I \geq 3) \\
    &+ \left(1 - \frac{1}{N_I}\right)\left(1 - \frac{3}{N_I}\right)\frac{1}{18}\mathbf{1}(N_I \geq 4) \\
    k_2 &\rightarrow \frac{1}{N_I}\left(1 - \frac{1}{N_I}\right) \left[\frac{1}{N_I} \mathbf{1}(N_I \geq 2) + \left(1 - \frac{2}{N_I}\right)\frac{1}{3}\mathbf{1}(N_I \geq 3) \right] \\
    &+\left(1 - \frac{1}{N_I}\right) \left[\frac{1}{3} \frac{1}{N_I} \mathbf{1}(N_I \geq 3) + \left(1 - \frac{3}{N_I}\right)\frac{1}{9}\mathbf{1}(N_I \geq 4)\right] \\
    k_3 & \rightarrow \left(1 - \frac{1}{N_I}\right) \left[\frac{1}{3} \frac{1}{N_I} \mathbf{1}(N_I \geq 3) + \left(1 - \frac{3}{N_I}\right)\frac{1}{9}\mathbf{1}(N_I \geq 4)\right].
\end{align*} The differences are \begin{align*}
k_1 - k_2 &= \begin{cases}
    1 & \textnormal{if} \ N_I = 1 \\
    \frac{3}{8} & \textnormal{if} \ N_I = 2 \\
    \frac{1}{27} & \textnormal{if} \ N_I = 3 \\
    \frac{1}{N_I^3} + \frac{1}{6N_I^2} - \frac{1}{9N_I} - \frac{1}{18} & \textnormal{if} \ N_I \geq 4,
\end{cases} \\
k_1 - k_3 &= \begin{cases}
    1 & \textnormal{if} \ N_I = 1 \\
    \frac{1}{4} & \textnormal{if} \ N_I = 2 \\
    \frac{5}{81} & \textnormal{if} \ N_I = 3 \\
    \frac{2}{3N_I^3} + \frac{1}{6N_I^2} + \frac{2}{9N_I} - \frac{1}{18} & \textnormal{if} \ N_I \geq 4.
\end{cases}
\end{align*}

To prove the Theorem, we observe Proposition \ref{prop:17} and compute the boundary of the anomaly zone by computing $\mathbf{E}'[k_1 - k_2]$ and $\mathbf{E}'[k_1 - k_3]$ using the modified geometric distribution of $N_I$ described in Section 3.  First, in the limit as $\eta_I$ and $\eta_J$ converge to $0$, the topology of the tree is essentially a star tree.  A star tree has a single interior vertex adjacent to all other vertices in the tree.  In Proposition \ref{prop:18}, we first compute the probability that each species has at least one surviving gene copy.

\begin{prop} \label{prop:18}
    In the limit as $\eta_I$ and $\eta_J$ converge to $0$, the unconditional survival probability converges to $$e^{-(\lambda - \mu)s}(1 - p_0(s))^2 \left\{ \frac{\alpha^4}{1 - \alpha^4 \beta} - 4 \frac{\alpha^3}{1 - \alpha^3 \beta} + 6 \frac{\alpha^2}{1 - \alpha^2 \beta} - 4 \frac{\alpha}{1 - \alpha \beta} + \frac{1}{1 -  \beta}\right\},$$ where $\alpha = \frac{\mu}{\lambda}q(h)$ and $\beta = q(s).$
\end{prop}

\begin{proof}
    Given $N_I = n$, the probability that $N_A = 0$ is $\left[\frac{\mu}{\lambda}q(h)\right]^{n}$.  That $N_A,N_B,N_C,N_D$ are independent conditioned on $N_I$ implies the probability of survival given $N_I  = n$ is $$\left(1 - \left[\frac{\mu}{\lambda}q(h)\right]^{n}\right)^{4}.$$ By the law of total probability, we have the unconditional probability of survival is \begin{align*}
        \sum_{n=1}^{\infty} \left(1 - \left[\frac{\mu}{\lambda}q(h)\right]^{n}\right)^{4} \mathbf{P}(N_I = n) = e^{-(\lambda - \mu)s}(1-p_0(s))^2 \sum_{n=1}^{\infty} \left(1 - \left[\frac{\mu}{\lambda}q(h)\right]^{n}\right)^{4} q(s)^{n-1}.
    \end{align*} Inside the series on the right-hand side, we expand to obtain summands of the form \begin{align*}
    \beta^{-1}\left\{(\alpha^4 \beta)^n - 4(\alpha^3 \beta)^n + 6(\alpha^2 \beta)^n - 4 (\alpha \beta)^n + \beta^n \right\}.
    \end{align*}  The result follows by summing the many geometric series.
\end{proof}

Finally, we give a proof of Theorem \ref{thm:caterpillar}.  It amounts to computing the expectations of $k_1 - k_2$ and $k_2 - k_3$ using measure $\mathbf{P}'$.  The unconditional probability found in Proposition \ref{prop:18} does appear in the exact expression for these expectations, and this expression was provided for completeness.  However, the specific form is not needed to show that $\mathbf{E}'[k_1 - k_2]$ can be negative, as required.

\begin{proof}[Proof of Theorem \ref{thm:caterpillar}]
The difference in probabilities is $$\frac{e^{-(\lambda - \mu)s}(1-p_0(s))^2}{\mathbf{P}(N_A,N_B,N_C,N_D > 0)} \sum_{n=1}^{\infty}(k_1 - k_j)(n) (1 - \alpha^n)^4 \beta^{n-1}, j = 2,3.$$ As $\lambda \rightarrow +\infty$, we find that $\alpha \rightarrow 0$ and $\beta \rightarrow 1$, so $\mathbf{P}(N_A,N_B,N_C,N_D > 0) \rightarrow 1$.  It remains to check whether the series is negative for sufficiently large $\lambda$. For all $n \geq 4$, we have $$\frac{1}{n^3} + \frac{1}{6n^2} - \frac{1}{9n} - \frac{1}{18} \leq -\frac{11}{192},$$ so the series is bounded above by \begin{equation} \label{eqn:k1-k2UB}(1-\alpha^4) + \frac{3}{8}(1 - \alpha^2)\beta + \frac{1}{27}(1-\alpha^3)^4 \beta^2 - \frac{11}{192}\sum_{n=4}^{\infty}(1-\alpha^n)^4\beta^{n-1}. \end{equation} Similarly to the proof of Proposition \ref{prop:18}, the series sums to $$\frac{\alpha^{16}\beta^3}{1-\alpha^{4} \beta} - 4 \frac{\alpha^{12}\beta^3}{1-\alpha^3 \beta} + 6 \frac{\alpha^{8}\beta^3}{1 - \alpha^2\beta}- 4 \frac{\alpha^{4}\beta^3}{1 - \alpha \beta} + \frac{\beta^3}{1-\beta}.$$ As $\lambda$ goes to $+\infty$, the expression in \eqref{eqn:k1-k2UB} converges to $-11/192$.  It follows then that $\mathbf{E}'[k_1-k_2]$ is bounded above by a function that is negative for sufficiently large $\lambda$, finding a zone where $((a,b),(c,d))$ has greater probability than the species topology.

For $k_1 - k_3$ given $N_I$, we note that $k_1 - k_3$ remains positive until $N_I$ is at least $6$.  For $n \geq 6$, we have $$\frac{2}{n^3} + \frac{1}{6n^2} + \frac{2}{9n} - \frac{1}{18} \leq -\frac{7}{648}.$$ Repeating the steps in the previous case shows there is $\lambda$ sufficiently large to provide a zone where $((a,c),(b,d))$ also has greater probability than the species topology.  This completes the proof of the Theorem. 
\end{proof}

\subsection{Computational results}

We plot the \textit{expected} values of $k_1-k_2$ and $k_1 - k_3$ for critical $\lambda$ in two cases of $\mu$ in the case where $\eta_I = \eta_J = 0$ with specific settings for the weights of the root and pendant edges, see Figure \ref{fig:boundary}.  Proposition \ref{prop:16} and Theorem \ref{thm:caterpillar} are both apparent from these results.  Notably, $\lambda$ must be quite large relative to $\mu$ to obtain an anomaly zone.

The results in Figure \ref{fig:boundary} show that anomaly zones of types (II) and (III) can be found for some settings of branch lengths in the case where $\mathcal{T}$ has the rooted caterpillar topology.  Case (I) of no anomaly zone occurs in the lower left region (left and below the left curve); Case (II) where $((a,b),(c,d))$ has maximal probability occurs in the middle region but the species tree has the second highest probability; and Case (III) where all balanced topologies have higher probability than the species tree.  Relative to the net speciation rate $\lambda - \mu$, the width of the regions associated to (II) are quite narrow.

\begin{figure}%
    \centering
    \subfloat[\centering]{{\includegraphics[width=6cm]{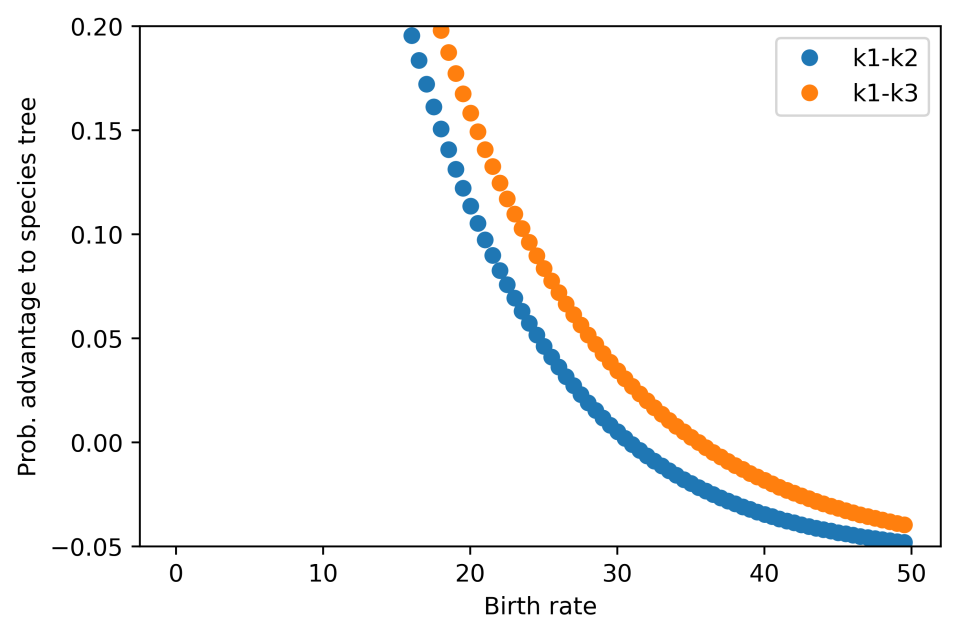} }}%
    \qquad
    \subfloat[\centering]{{\includegraphics[width=6cm]{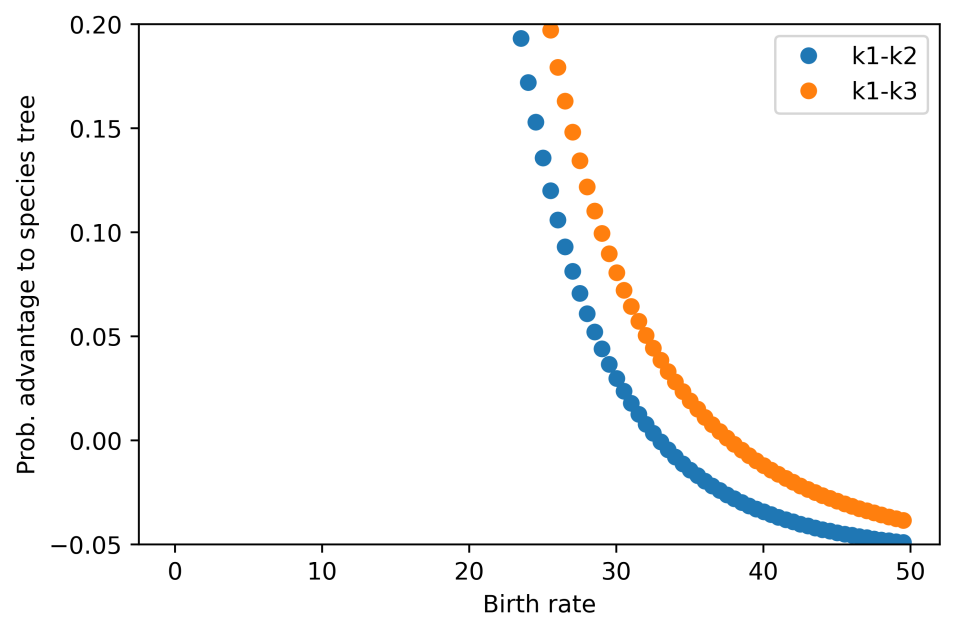} }}%
    \caption{(a) Expected values of $k_1-k_2$ and $k_2 - k_3$ in the case of death rate $\mu = 0.01$, root edge of length $0.01$, and leaf edges all length $0.05$.  (b) Same settings as in (a), except the death rate is $\mu = 3$. }%
    \label{fig:boundary}%
\end{figure}

}

\section{Conclusion}

Through more careful counting and bounding compared to previous efforts in this area, we showed that rooted balanced species quartets have no anomaly zones.  \rev{We also showed that} if anomaly zones exist, we have shown the respective anomalous gene trees must be balanced quartets.   The statements of various Propositions also provide a partial ranking of uniformly sampled gene tree topologies, in results somewhat analogous to \cite{allman2011identifying}.\rev{Moreover, the rooted caterpillar topology on four leaves has branch length settings, provided the birth rate is sufficiently large.  As with the MSC, anomalous gene trees occur when both the interior branch lengths approach $0$. It is not clear what maximal branch lengths provide existence of anomaly zones.  Because the anomaly zones in Section 5.2 required a much larger birth rate $\lambda$ than the death rate, we can hypothesize that GDL anomaly zones may exist, but they are fairly remote.  The utilized birth and death parameters in the fungal data set of \cite{rasmussen2012unified} and the simulation study of \cite{yan2022} are comparable to each other, i.e. the birth rate is taken to be close if not equal to the death rate.  So, it might be expected that anomaly zones for the caterpillar tree are not an important confounding factor, given their remoteness.}

This analysis shows that GDL could have similar issues with gene tree discordance as has been observed with MSC, but the structural information contained in the species tree might be more easily ascertained.  This is because the branching events at speciation points must always occur.  Because of possible anomalous gene trees provided in Theorem 3, it should be expected that this problem only becomes harder as more species are added to the tree.  However, the analysis gets out of hand quickly as more gene trees are possible.  It should be expected that only experimental or simulation evidence is practical to obtain in the case of larger trees.

Because the bounds obtained in Sections 5 and 6 utilize only the \rev{numbers of copies} at particular points in the tree, it should be expected that the results of this paper and those of, e.g. \cite{legriedGDL2021,eulenstein2020unified}, etc. generalize to the case where birth and death rates are taken to vary over time or across edges in the species tree.  There are practical issues raised about how to estimate the rate parameters from data \cite{louca2020extant}, \cite{legried2021birthdeath}, and \cite{legried2021IdentifiabilityInference}. The difficulties raised there do not translate here:  \rev{error-free gene trees are sufficient to recover the species tree.}

\rev{Moreover, the reliance on population sizes rather than specific sequences suggests that there are potential generalizations of these results to other models.  Even if the specific gene sequences engage in \textit{concerted evolution} (see \cite{velandia2016}) where sequences may not evolve independently of each other, the species tree may still be recoverable without estimating the gene tree.  In this setting, the gene tree is not easy to estimate because copies of a given species look so closely related to make distance-based reconstruction difficult.  However, the method of ASTRAL-one really only requires a single copy of each species, so there may be hope of at least reconstructing the species tree.  One could then attempt to perform gene tree reconciliation such as in \cite{rasmussen2012unified} to characterize branching events, though a full characterization of the gene tree branching events may still be impossible.}

\rev{A related unresolved question is the estimation of internal branch lengths of the species tree.  For the MSC, \cite{liu2010} found a method of estimating the species tree divergence time for a rooted triple using a maximum pseudo-likelihood.  By counting how frequently each gene tree topology appears, the percentage of gene trees $p$ following the species tree coincides with the theoretical probability $p = 1 - \frac{2}{3}e^{-\omega}$ where $\omega$ is the interior branch length.  Solving for $p$ gives a consistent estimator of the interior branch length.  A similar result could be proved for GDL, but we expect the inverse problem is not tractable.  One could seek a numerical solution and give a formal proof of the existence of a unique solution to such an equation.  The question becomes harder when there is gene tree estimation error, though one could attempt to extend the results of \cite{roch2015} to GDL with methods used here.}

The analysis of uniformly sampled gene trees as inputs to ASTRAL-one is already complicated, but that does not imply similar difficulties will appear with ASTRAL-multi, \cite{rabiee2019multi}.  In ASTRAL-multi, the multi-labelled gene tree is instead replaced with a singly-labelled gene tree for every choice of species in the gene tree.  The dependencies between same-labelled individuals in a gene tree make analysis seem daunting, but one advantage of ASTRAL-multi is that there is no need to condition on survival of all species.

A disadvantage of conditioning on a present-day observation is that it induces a survivorship bias that can be difficult to model.  In this paper, we engaged in analysis that avoids having to consider this bias.  \rev{In particular, the anomaly zone results in this paper does not require an explicit calculation of the probability of coalescence of two copies from different species, given survival.} Another possibility is to analyze the distribution of pseudoorthologs only, see \cite{smith2022}.  This approach would require new methods, as one conditions \rev{on an explicit pattern of duplications and losses that yields a single-copy gene tree.  This conditional distribution seems more difficult to work with.}

Lastly, the results of this paper suggest but do not give proof of existence \rev{or} non-existence of anomaly zones under the DLCoal model \cite{rasmussen2012unified}.  Rigorous proof would come through a more extensive analysis of the branching patterns than we can reasonably do here.  \rev{Even if balanced quartets have no anomaly zones due to GDL, there are still anomaly zones under MSC.  It is not clear whether the anomaly zone for DLCoal is identical to that of MSC.  For instance}, the resulting gene tree could have so many branches in it that the anomaly zone for the coalescence portion of the model is even larger than that of the simple MSC rooted quartet.  

\section{Acknowledgments}

This work was completed in part while the author was a postdoctoral fellow at the Department of Statistics at the University of Michigan-Ann Arbor, in addition to their current position at the Georgia Institute of Technology.  BL was supported by NSF grant DMS-1646108 and NSF-Simons grant for the Southeast Center for Mathematics and Biology DMS-1764406.

\bibliography{arXivV3.bib}

\begin{thebibliography}{32}
\providecommand{\natexlab}[1]{#1}
\providecommand{\url}[1]{\texttt{#1}}
\expandafter\ifx\csname urlstyle\endcsname\relax
  \providecommand{\doi}[1]{doi: #1}\else
  \providecommand{\doi}{doi: \begingroup \urlstyle{rm}\Url}\fi

\bibitem[Allman et~al.(2011)Allman, Degnan, and Rhodes]{allman2011identifying}
Elizabeth~S Allman, James~H Degnan, and John~A Rhodes.
\newblock Identifying the rooted species tree from the distribution of unrooted gene trees under the coalescent.
\newblock \emph{Journal of Mathematical Biology}, 62\penalty0 (6):\penalty0 833--862, 2011.
\newblock \doi{10.1007/s00285-010-0355-7}.

\bibitem[Arvestad et~al.(2009)Arvestad, Lagergren, and Sennblad]{ArLaSe:09}
Lars Arvestad, Jens Lagergren, and Bengt Sennblad.
\newblock The gene evolution model and computing its associated probabilities.
\newblock \emph{Journal of the ACM}, 56\penalty0 (2):\penalty0 7, 2009.
\newblock \doi{10.1145/1502793.1502796}.

\bibitem[Degnan and Rosenberg(2006)]{degnan2006}
J.~H. Degnan and N.~A. Rosenberg.
\newblock {{D}iscordance of species trees with their most likely gene trees}.
\newblock \emph{PLoS Genet.}, 2\penalty0 (5):\penalty0 e68, May 2006.
\newblock \doi{10.1371/journal.pgen.0020068}.

\bibitem[Durrett(2010)]{durrett_2010}
Rick Durrett.
\newblock \emph{Probability: Theory and Examples}.
\newblock Cambridge Series in Statistical and Probabilistic Mathematics. Cambridge University Press, 4 edition, 2010.
\newblock \doi{10.1017/CBO9780511779398}.

\bibitem[Felsenstein(2003)]{felsenstein2003inferring}
J.~Felsenstein.
\newblock \emph{Inferring Phylogenies}.
\newblock Sinauer, 2003.
\newblock ISBN 9780878931774.
\newblock URL \url{https://books.google.com/books?id=GI6PQgAACAAJ}.

\bibitem[Gascuel(2005)]{gascuel2005mathematics}
O.~Gascuel.
\newblock \emph{Mathematics of Evolution and Phylogeny}.
\newblock OUP Oxford, 2005.
\newblock ISBN 9780198566106.
\newblock URL \url{https://books.google.com/books?id=VjA8ThtLs7IC}.

\bibitem[Hill et~al.(2022)Hill, Legried, and Roch]{hill2022species}
Max Hill, Brandon Legried, and Sebastien Roch.
\newblock Species tree estimation under joint modeling of coalescence and duplication: sample complexity of quartet methods.
\newblock \emph{The Annals of Applied Probability}, 2022.
\newblock \doi{10.48550/ARXIV.2007.06697}.

\bibitem[Kendall(1948)]{kendall1948}
David~G. Kendall.
\newblock On the generalized birth-and-death process.
\newblock \emph{Ann. Math. Statist.}, 19\penalty0 (1):\penalty0 1--15, 03 1948.
\newblock \doi{10.1214/aoms/1177730285}.
\newblock URL \url{https://doi.org/10.1214/aoms/1177730285}.

\bibitem[Legried and Terhorst(2022)]{legried2021birthdeath}
Brandon Legried and Jonathan Terhorst.
\newblock A class of identifiable phylogenetic birth-death models.
\newblock \emph{Proceedings of the National Academy of Sciences}, 119\penalty0 (35):\penalty0 e2119513119, August 2022.
\newblock \doi{10.1073/pnas.2119513119}.

\bibitem[Legried and Terhorst(2023)]{legried2021IdentifiabilityInference}
Brandon Legried and Jonathan Terhorst.
\newblock Identifiability and inference of phylogenetic birth-death models.
\newblock \emph{Journal of Theoretical Biology}, 2023.
\newblock \doi{10.1016/j.jtbi.2023.111520}.

\bibitem[Legried et~al.(2021)Legried, Molloy, Warnow, and Roch]{legriedGDL2021}
Brandon Legried, Erin Molloy, Tandy Warnow, and Sebastien Roch.
\newblock Polynomial-time statistical estimation of species trees under gene duplication and loss.
\newblock \emph{Journal of Computational Biology}, 28:\penalty0 452--468, 2021.
\newblock \doi{10.1089/cmb.2020.0424}.

\bibitem[Liu et~al.(2010)Liu, Yu, and Edwards]{liu2010}
Liang Liu, Lili Yu, and Scott Edwards.
\newblock A maximum pseudo-likelihood approach for estimating species trees under the coalescent model.
\newblock \emph{BMC Evolutionary Biology}, 10, 2010.
\newblock \doi{https://doi.org/10.1186/1471-2148-10-302}.

\bibitem[Louca and Pennell(2020)]{louca2020extant}
Stilianos Louca and Matthew~W Pennell.
\newblock Extant timetrees are consistent with a myriad of diversification histories.
\newblock \emph{Nature}, 580\penalty0 (7804):\penalty0 502--505, April 2020.
\newblock ISSN 0028-0836, 1476-4687.
\newblock \doi{10.1038/s41586-020-2176-1}.

\bibitem[Markin and Eulenstein(2021)]{eulenstein2020unified}
Alexey Markin and Oliver Eulenstein.
\newblock Quartet-based inference is statistically consistent under the unified duplication-loss-coalescence model.
\newblock \emph{Bioinformatics}, 37\penalty0 (22):\penalty0 4064--4074, 2021.
\newblock \doi{10.1093/bioinformatics/btab414}.

\bibitem[Mirarab et~al.(2014)Mirarab, Reaz, Bayzid, Zimmermann, Swenson, and Warnow]{astral}
S.~Mirarab, R.~Reaz, Md.~S. Bayzid, T.~Zimmermann, M.~S. Swenson, and T.~Warnow.
\newblock {ASTRAL: genome-scale coalescent-based species tree estimation}.
\newblock \emph{Bioinformatics}, 30\penalty0 (17):\penalty0 i541--i548, 2014.
\newblock \doi{10.1093/bioinformatics/btu462}.

\bibitem[Molloy and Warnow(2020)]{molloy2020}
Erin Molloy and Tandy Warnow.
\newblock Fastmulrfs: fast and accurate species tree estimation under generic gene duplication and loss models.
\newblock \emph{Bioinformatics}, 36, 2020.
\newblock \doi{10.1093/bioinformatics/btaa444}.

\bibitem[Rabiee et~al.(2019)Rabiee, Sayyari, and Mirarab]{rabiee2019multi}
Maryam Rabiee, Erfan Sayyari, and Siavash Mirarab.
\newblock Multi-allele species reconstruction using {ASTRAL}.
\newblock \emph{Molecular Phylogenetics and Evolution}, 130:\penalty0 286--296, 2019.
\newblock \doi{10.1016/j.ympev.2018.10.033}.

\bibitem[Rannala and Yang(2003)]{rannala2003bayes}
Bruce Rannala and Ziheng Yang.
\newblock Bayes estimation of species divergence times and ancestral population sizes using dna sequences from multiple loci.
\newblock \emph{Genetics}, 164\penalty0 (4):\penalty0 1645--1656, 2003.
\newblock \doi{10.1093/genetics/164.4.1645}.

\bibitem[Rasmussen and Kellis(2012)]{rasmussen2012unified}
M.~D. Rasmussen and M.~Kellis.
\newblock {Unified modeling of gene duplication, loss, and coalescence using a locus tree}.
\newblock \emph{Genome Research}, 22\penalty0 (4):\penalty0 755--765, 2012.
\newblock \doi{10.1101/gr.123901.111}.

\bibitem[Roch and Snir(2013)]{roch12lateral}
Sebastien Roch and Sagi Snir.
\newblock Recovering the treelike trend of evolution despite extensive lateral genetic transfer: A probabilistic analysis.
\newblock \emph{Journal of Computational Biology}, 20\penalty0 (2):\penalty0 93--112, 2015/06/08 2013.
\newblock \doi{10.1089/cmb.2012.0234}.
\newblock URL \url{http://dx.doi.org/10.1089/cmb.2012.0234}.

\bibitem[Roch and Steel(2015)]{RochSteel:15}
Sebastien Roch and Mike Steel.
\newblock Likelihood-based tree reconstruction on a concatenation of aligned sequence data sets can be statistically inconsistent.
\newblock \emph{Theoretical Population Biology}, 100:\penalty0 56 -- 62, 2015.
\newblock ISSN 0040-5809.
\newblock \doi{http://dx.doi.org/10.1016/j.tpb.2014.12.005}.
\newblock URL \url{http://www.sciencedirect.com/science/article/pii/S0040580914001075}.

\bibitem[Roch and Warnow(2015)]{roch2015}
Sebastien Roch and Tandy Warnow.
\newblock On the robustness to gene tree estimation error (or lack thereof) of coalescent-based species tree methods.
\newblock \emph{Systematic Biology}, 64, 2015.
\newblock \doi{10.1093/sysbio/syv016}.

\bibitem[Semple and Steel(2003)]{SempleSteel:03}
C.~Semple and M.~Steel.
\newblock \emph{Phylogenetics}, volume~22 of \emph{Mathematics and its Applications series}.
\newblock Oxford University Press, 2003.

\bibitem[Shekhar et~al.(2018)Shekhar, Roch, and Mirarab]{Shekhar_2018}
Shubhanshu Shekhar, Sebastien Roch, and Siavash Mirarab.
\newblock Species tree estimation using astral: How many genes are enough?
\newblock \emph{IEEE/ACM Transactions on Computational Biology and Bioinformatics}, 15\penalty0 (5):\penalty0 1738–1747, Sep 2018.
\newblock ISSN 2374-0043.
\newblock \doi{10.1109/tcbb.2017.2757930}.
\newblock URL \url{http://dx.doi.org/10.1109/TCBB.2017.2757930}.

\bibitem[Smith and Hahn(2022)]{smith2022}
Megan Smith and Matthew Hahn.
\newblock The frequency and topology of pseudoorthologs.
\newblock \emph{Systematic Biology}, 71:\penalty0 649--659, 2022.
\newblock \doi{10.1093/sysbio/syab097}.

\bibitem[Steel(2016)]{steelbook2016}
Mike Steel.
\newblock \emph{Phylogeny: Discrete and Random Processes in Evolution}.
\newblock SIAM-Society for Industrial and Applied Mathematics, Philadelphia, PA, USA, 2016.
\newblock ISBN 161197447X.

\bibitem[Tavare(1986)]{tavare1986seq}
Simon Tavare.
\newblock Some probabilistic and statistical problems in the analysis of dna sequences.
\newblock \emph{Lectures on Mathematics in the Life Sciences}, 17, 1986.

\bibitem[Velandia-Huerto et~al.(2016)Velandia-Huerto, Berkemer, Hoffman, Retzlaff, Marroquin, Hernandez-Rosales, Stadler, and Bermudez-Santana]{velandia2016}
Cristian Velandia-Huerto, Sarah Berkemer, Anne Hoffman, Nancy Retzlaff, Liliana~Romero Marroquin, Maribel Hernandez-Rosales, Peter Stadler, and Clara Bermudez-Santana.
\newblock Orthologs, turn-over, and remolding of trnas in primates and fruit flies.
\newblock \emph{BMC Genomics}, 17, 2016.
\newblock \doi{10.1186/s12864-016-2927-4}.

\bibitem[Warnow(2017)]{warnow2017book}
Tandy Warnow.
\newblock \emph{Computational Phylogenetics: An Introduction to Designing Methods for Phylogeny Estimation}.
\newblock Cambridge University Press, 2017.
\newblock \doi{10.1017/9781316882313}.

\bibitem[Yan et~al.(2022)Yan, Smith, Du, Hahn, and Nakhleh]{yan2022}
Zhi Yan, Megan Smith, Peng Du, Matthew Hahn, and Luay Nakhleh.
\newblock Species tree inference methods intended to deal with incomplete lineage sorting are robust to the presence of paralogs.
\newblock \emph{Systematic Biology}, 71, 2022.
\newblock \doi{10.1093/sysbio/syab056}.

\bibitem[Yang(2014)]{yang2014molecular}
Z.~Yang.
\newblock \emph{Molecular Evolution: A Statistical Approach}.
\newblock OUP Oxford, 2014.
\newblock ISBN 9780191023309.
\newblock URL \url{https://books.google.com/books?id=T-LoAwAAQBAJ}.

\bibitem[Zhang et~al.(2020)Zhang, Scornavacca, Molloy, and Mirarab]{zhang2020}
Chao Zhang, Celine Scornavacca, Erin Molloy, and Siavash Mirarab.
\newblock Astral-pro: Quartet-based species-tree inference despite paralogy.
\newblock \emph{Molecular Biology and Evolution}, 37, 2020.
\newblock \doi{10.1093/molbev/msab232}.

\end{thebibliography}

\rev{

\section*{Appendix}

\subsection*{Numbers of copies}

The first two raw moments are easily computed by summing the geometric series.  The first moment (the mean) is \begin{align*}
    \EE[N_s] &= e^{-(\lambda - \mu)s}\left(1-p_0(s)\right)^2 \sum_{i=1}^{\infty}i q(s)^{i-1} \\
    &= e^{-(\lambda - \mu)s}\left(1-p_0(s)\right)^2 \frac{1}{\left(1-q(s)\right)^2} \\
    &= e^{-(\lambda - \mu)s} \left(\frac{\lambda - \mu q(s)}{\lambda - \lambda q(s)} \right)^2.
\end{align*} The second moment is \begin{align*}
    \EE[N_s^2] &= e^{-(\lambda - \mu)s}\left(1 - p_0(s)\right)^2 \sum_{i=1}^{\infty}i^2 q(s)^{i-1} \\
    &= e^{-(\lambda - \mu)s}\left(1-p_0(s)\right)^2 \frac{1+q(s)}{\left(1-q(s)\right)^3} \\
    &= e^{-(\lambda - \mu)s} \frac{1+q(s)}{1-q(s)} \left(\frac{\lambda - \mu q(s)}{\lambda - \lambda q(s)}\right)^2.
\end{align*} The variance is then computed as $$\Var[N_s] = \EE[N_s^2] - \left(\EE[N_s]\right)^2.$$

One may be interested in the conditional distribution of $N_s$, conditioned on non-extinction by time $s$, i.e. $N_s > 0$.  Because $$\PP(N_s > 0) = 1 - \PP(N_s = 0) = 1 - p_0(s),$$ it follows for any $i > 0$ that $$\PP(N_s = i|N_s > 0) = \frac{\PP(N_s = i)}{\PP(N_s > 0)} = 
e^{-(\lambda - \mu)s}\left(1-p_0(s)\right)q(s)^{i-1}.$$ So $N_s|N_s > 0$ is a geometric random variable with success probability $1-q(s)$.  The mean is $$\EE[N_s|N_s > 0] = \frac{1}{1-q(s)}$$ and the variance is $$\Var[N_s|N_s > 0] = \frac{1+q(s)}{\left(1-q(s)\right)^2} -\frac{1}{\left(1-q(s)\right)^2} = \frac{q(s)}{\left(1-q(s)\right)^2}.$$

Next, the moment generating function can be used to compute some useful statistics. We have \begin{align*}
    M_s(\tau) = \EE[e^{\tau N_s}] &= \frac{\mu}{\lambda}q(s) + e^{-(\lambda - \mu)s}\left(1-p_0(s)\right)^2 \sum_{i=0}^{\infty}e^{\tau i} q(s)^{i-1} \\
    &= \frac{\mu}{\lambda}q(s) + e^{-(\lambda - \mu)s}\left(1 - p_0(s)\right)^2 \frac{e^\tau}{ 1 - q(s)e^{\tau}}
\end{align*} This function is defined for all $x$ such that $1 - e^\tau q(s) > 0$, which is an open neighborhood containing the origin.

The negative raw moments can also be computed.  More generally, consider $N_s^{-P}e^{\tau N_s}$ for integers $P \geq 1$, conditioned on survival.  The $P=1$ case has a closed-form expression: \begin{align*}
    \EE\left[\frac{1}{N_s} e^{\tau N_s} | N_s > 0\right] &= e^{-(\lambda - \mu)s} \left(1 - p_0(s)\right) \sum_{i=1}^{\infty}\frac{1}{i} q(s)^{i-1} \\
    &= e^{-(\lambda - \mu)s} \frac{1-p_0(s)}{q(s)} \int \sum_{i=1}^{\infty}q(s)^{i} \ dq(s) \\
    &= e^{-(\lambda - \mu)s} \frac{1-p_0(s)}{q(s)} \log\left(\frac{1}{1 - q(s)e^\tau}\right).
\end{align*} For greater integers $P$, there is no closed form expression, but the well-known polylogarithm function can be used.  It is defined as $$\Li_{P}(\zeta) = \sum_{i=1}^{\infty} \frac{\zeta^i}{i^P}.$$ Then $$\EE\left[\frac{1}{N_s^P} e^{\tau N_s}| N_s > 0\right] = e^{-(\lambda - \mu)s} \frac{1-p_0(s)}{q(s)}\Li_{P}\left(q(s)e^{\tau}\right).$$

\begin{proof}
For the denominator on the left-hand side, we first start with the fact that $P^{-2} = \int_{\tau=0}^{\infty}\tau e^{-\tau P} \ d\tau$.  Then $$\frac{1}{\left(\sum_{j=1}^{N_I}N_{I,j}\right)^2} = \int_{\tau=0}^{\infty}\tau e^{-\tau\sum_{j=1}^{N_I}N_{I,j}} \ d\tau.$$ The $N_{I,j}$ are identically distributed, so $$\mathbf{E}'\left[\frac{\sum_{j=1}^{N_I}N_{I,j}^2}{\left(\sum_{j=1}^{N_I}N_{I,j}\right)^2} \bigg| N_I \right] = N_I \mathbf{E}'\left[\frac{N_{I,1}^2}{\left(\sum_{j=1}^{N_I}N_{I,j}\right)^2} \bigg |N_I \right].$$ To utilize independence, we separate the $N_{I,1}$ factors from the others.  We have \begin{equation*}
            \mathbf{E}'\left[\frac{N_{I,1}^2}{\left(\sum_{j=1}^{N_I}N_{I,j}\right)^2} \bigg |N_I \right] = \int_{\tau=0}^{\infty} \tau \EE' \left[N_{I,1}^2 e^{-\tau N_{I,1}} | N_I \right] \mathbf{E}'\left[e^{-\tau \sum_{j=2}^{N_{I}} N_{I,j}}| N_I \right] \ dx.
\end{equation*} A couple basic properties of the moment generating function are that $M_{N,1}''(-\tau) = \EE'\left[N_{I,1}^2 e^{-\tau N_{I,1}}\right]$ and $M_{N_{I,1}}(-\tau)^{N_I - 1} = \prod_{j=2}^{N_I}\EE'[e^{-\tau N_{I,1}}|N_I] = \EE'[e^{-\tau \sum_{j=2}^{N_I}N_{I,j}}|N_I]$.  Plugging these into the integral expression completes the proof. 
\end{proof}

}

\subsection*{Other calculations for species trees on three leaves}

Throughout this section, the species tree $\sigma = (\mathcal{T},\mathbf{f})$ is assumed to have three leaf species $A,B$, and $C$ and the topology has $A,B$ are siblings with $C$ as the outgroup.  This rooted topology can be expressed using the Newick tree format:  $((A,B),C)$.  The edge lengths are $s$ for the stem edge, $f$ for the interior edge, and $g$ for the pendant edges incident to $A$ and $B$.  The tree is ultrametric, so the pendant edge incident to $C$ has length $f+g$.  The root vertex is labeled $I$, and the parent vertex to $A$ and $B$ is labeled $J$.  For any vertex $v$ in the species tree, let $N_v$ be the size of the population.  Our problem setting requires us to assume every contemporary species has at least one copy in a given gene tree, meaning $N_A, N_B,$ and $N_C$ are positive.  An example of a possible gene tree is given in Figure 1.  Let $\PP'$ be the probability measure subject to this conditioning.  Expectations computed under the respective measures are denoted $\mathbf{E}$ and $\mathbf{E}'$.

In Proposition \ref{prop:1}, we give the transition probability for the continuous-time process $\{N_s\}_{s \geq 0}$.

\begin{prop}\label{prop:1}
Let $s,f > 0$ so that $s$ and $s+f$ are successive times.  Then $$\PP(N_{s+f} = j|N_s = i) = p_{j}(f)^{(\ast i)},$$ where $p_j(\cdot)^{(\ast i)}$ is the $i$-fold convolution of the function $p_j$.
\end{prop}

\begin{proof}
The progenies of the $i$ individuals at time $s$ are independent and identically distributed random variables, and $N_{s+f}$ is the sum.  So $N_{s+f}$ has the same distribution as $\sum_{j=1}^{i}N_f$.  The distribution of the independent and identically distributed summands is known to be the discrete convolution of $i$ copies of $p_{j}(\cdot)$ evaluated at $f$.  
\end{proof}

In Proposition \ref{prop:2}, we write the survival probability as an expectation over the numbers of copies at the internal nodes of $\mathcal{T}$.

\begin{prop}\label{prop:2}
We have \begin{align*}\PP(N_A,N_B,N_C > 0) &= \EE\left[\left(1 - p_0(g)^{N_J}\right)^2 \left(1 - p_0(f+g)^{N_I}\right)\right]\\
&= \EE\left[\EE\left[\left(1 - p_0(g)^{N_J}\right)^2 | N_I\right] \left(1 - p_0(f+g)^{N_I}\right)\right].\end{align*}
\end{prop}

\begin{proof}
The second equality is immediate by the definition of conditional expectation. 
The rest of the proof is for the first equality.  Condition on the specified values of $N_I$ and $N_J$.  We first have $\PP(N_I = i) = p_i(t)$ and $\PP(N_J = j|N_I = i) = p_{\bullet}(f)$.  The Markov property and conditional independence imply \begin{align*}&\PP(N_A,N_B,N_C > 0) \\
&= \sum_{i=1}^{\infty}\PP(N_A,N_B,N_C > 0|N_I = i)\PP(N_I = i) \\
&= \sum_{i=1}^{\infty}\PP(N_A, N_B > 0|N_I = i)\PP(N_C > 0|N_i = i) \PP(N_I = i)\\
&= \sum_{i=1}^{\infty} \sum_{j=1}^{\infty}\PP(N_A, N_B > 0|N_J = j, N_I = i) \\
&\ \ \ \ \ \ \ \ \ \ \ \ \ \times \PP(N_J = j|N_I = i)\PP(N_C > 0|N_I = i) \PP(N_I = i).\end{align*}  Using the modified geometric mass function and the transition probability from Proposition \ref{prop:1}, the right-hand side equals \begin{align*}
&\sum_{i=1}^{\infty}\sum_{j=1}^{\infty}\PP(N_A, N_B > 0|N_J = j)\PP(N_J = j|N_I = i) \PP(N_C > 0 |N_I = i)\PP(N_I = i) \\
&= \sum_{i=1}^{\infty} \sum_{j=1}^{\infty}\PP(N_A > 0|N_J = j)\PP(N_B > 0|N_J = j)\PP(N_C > 0|N_I = i)p_i(t)p_{j}(f)^{(\ast i)} \\
&=\sum_{i=1}^{\infty}\sum_{j=1}^{\infty} \left(1 - p_0(g)^j\right)^2 \left(1 - p_0(f+g)^i\right) p_i(t)p_{j}(f)^{(\ast i)} \\
&= \EE\left[\left(1 - p_0(g)^{N_J}\right)^2 \left(1 - p_0(f+g)^{N_I}\right) \right]. 
\end{align*}
\end{proof}

In Proposition \ref{prop:3}, we compute the conditional expectation of any functional of the population size at $I$.

\begin{prop}\label{prop:3}
We have $$\EE'[h(N_I)] = \frac{\EE\left[h(N_I) \EE\left[\left(1 - p_0(g)^{N_J}\right)^2 |N_I\right] \left(1 - p_0(f+g)^{N_I}\right) \right]}{\PP(N_A, N_B, N_C > 0)}.$$ In the numerator, the outside expectation is subject to the distribution $p_i(t)$ and the inside expectation is subject to the distribution $p_j(f)^{(\ast N_I)}$.
\end{prop}

\begin{proof}
We use Bayes' Rule to unpack the conditional expectation.  Observe that $\PP(N_I = 0|N_A, N_B, N_C > 0) = 0$.  Then \begin{align*}
    \EE'[h(N_I)] &= \sum_{i=1}^{\infty}h(i) \PP(N_I = i|N_A, N_B, N_C > 0) \\
    &= \sum_{i=1}^{\infty} h(i)\frac{\PP(N_A, N_B, N_C > 0|N_I = i)}{\PP(N_A,N_B,N_C > 0)}\PP(N_I = i) \\
    &= \sum_{i=1}^{\infty} h(i)\frac{\PP(N_A, N_B > 0|N_I = i)\PP(N_C > 0|N_I = i)}{\PP(N_A, N_B, N_C > 0)}\PP(N_I = i) \\
    &= \frac{\EE\left[h(N_I) \EE\left[\left(1 - p_0(g)^{N_J}\right)^2 |N_I\right] \left(1 - p_0(f+g)^{N_I}\right) \right]}{\PP(N_A, N_B, N_C > 0)}.
\end{align*}
\end{proof} 

Conditioned on $N_I$, for any $j \in \{1,2,...,N_I\}$ let $N_{I,j} \geq 0$ be the size of the progeny of individual $i$ that survives to $J$.  We now expand on the terms developed in Propositions \ref{prop:2} and \ref{prop:3}.

\begin{prop} \label{prop:4}
We have $$\EE\left[\left(1 - p_0(g)^{N_J}\right)^2 \bigg|N_I \right] = 1 - 2 \left\{M_f\left[\log(p_0(g) )\right]\right\}^{N_I} + \left\{M_f \left[2 \log(p_0(g) ) \right]\right\}^{N_I}.$$
\end{prop}

\begin{proof}
Given $N_I$, the random variable $N_J$ is the sum of $N_I$ iid random variables $N_{I,j}$ taking the same modified geometric distribution with branch length $f$.  Recall the distribution of $N_{I,j}$ is the same as $N_f$, for the number of progeny of a single individual at the end of a branch of length $f$. The left-hand side then equals \begin{align*}&1 - 2 \left( \EE\left[p_0(g)^{N_f}\right] \right)^{N_I} + \left(\EE\left[ \{p_0(g)^2\}^{N_f}\right] \right)^{N_I} \\
&= 1 - 2 \left\{ M_f \left[ \log(p_0(g)) \right] \right\}^{N_I} + \left\{M_f \left[ \log(p_0(g)^2)\right]\right\}^{N_I}.\end{align*}
\end{proof}

From here, a more explicit expression for $\PP(N_A, N_B, N_C > 0)$ is revealed.  We let $a_3 = M_f\left(\log(p_0(g))\right)$, $b_3 = M_f\left(2 \log(p_0(g)) \right)$ and $c_3 = p_0(f+g)$.

\begin{prop} \label{prop:5}
We have \begin{align*}\PP(N_A, N_B, N_C > 0) &= 1 - 2 M_t(\log a_3) + M_t(\log b_3) \\
&- M_t(\log c_3) + 2 M_t(\log a_3 + \log c_3) \\
&- M_t(\log b_3 + \log c_3).\end{align*}
\end{prop}

\begin{proof}
Combining Propositions \ref{prop:2} and \ref{prop:4}, we have \begin{align*}
    \PP(N_A,N_B,N_C > 0) &= \EE \left[ (1 - 2a_3^{N_I} + b_3^{N_I}) \left(1 - c_3^{N_I}\right) \right].
\end{align*}   Expanding the right-hand side implies \begin{align*}&1 - 2 \EE[a_3^{N_I}] + \EE[b_3^{N_I}] - \EE[c_3^{N_I}] + 2 \EE[ \left(a_3 c_3\right)^{N_I}] - \EE[(b_3 c_3)^{N_I}] \\
&= 1 - 2 \EE[e^{N_I \log a_3}] + \EE[e^{N_I \log b_3}] - \EE[e^{N_I \log c_3}] + 2 \EE[e^{N_I \log(a_3 c_3)}] - \EE[e^{N_I \log(b_3 c_3)}].\end{align*} Finally, $N_I$ has the same distribution as $N_t$, so it has the same moment generating function.
\end{proof}

In the next Proposition, we expand the numerator of $\EE'[h(N_I)]$.

\begin{prop} \label{prop:6}
The numerator in the statement of Proposition \ref{prop:3} simplifies to $$\EE\left[h(N_I) \left\{ (1 - 2 \{M_f [\log(p_0(g))] \}^{N_I} + \left\{M_f [2 \log(p_0(g)) ] \right\}^{N_I})(1 - p_0(f+g)^{N_I}) \right\}\right].$$
\end{prop}

\begin{proof}
This is immediate by using Proposition \ref{prop:4}.
\end{proof}

These Propositions can be applied to any non-negative function $h$, though the most natural choice is to use $$h(N_I) = \EE'\left[ \frac{\sum_{j=1}^{N_I}N_{I,j}^2}{\left(\sum_{j=1}^{N_I}N_{I,j}\right)^2} \bigg| N_I \right]$$ from the statement of Theorem 1. However, in \cite{legriedGDL2021}, interesting results were obtained in the case where $h(N_I) = 1/N_I^2$, which is a lower bound on the excessive weight associated to the species tree topology.  Recall $\mathcal{T}$ is the species tree topology, and let $\mathcal{T}'$ be either of the two alternative topologies.  Then we have $$\PP'(u(t) = \mathcal{T}) - \PP'(u(t) = \mathcal{T}') \geq \EE'\left[\frac{1}{N_I^2}\right],$$ by Lemma 1 of \cite{legriedGDL2021}.  This observation was important to show that the rooted species tree on three leaves is identifiable from gene trees generated by GDL, the unrooted species tree on four leaves is identifiable from unrooted gene trees generated by gene trees generated by GDL, and that ASTRAL-one is a statistically consistent estimator of any unrooted species tree on any fixed number of leaves as the number of gene trees goes to infinity.  However, understanding the distribution of $u(t)$ better requires computing $\EE'[h(N_I)]$ explicitly.  This work is left open here.

\end{document}